\documentclass{amsart}

\usepackage{graphicx}
\usepackage{amssymb,amsmath, amsfonts,amsthm}
\usepackage{relsize}
\def\proof{\noindent\hspace{2em}{\itshape Proof: }}
\def\QEDclosed{\mbox{\rule[0pt]{1.3ex}{1.3ex}}} 

\def\QED{\QEDclosed} 
\def\endproof{\hspace*{\fill}~\QED\par\endtrivlist\unskip}

\input xy 
\xyoption{all} 

\usepackage{color}
\newtheorem{theorem}{Theorem}
\newtheorem{definition}[theorem]{Definition}
\newtheorem{proposition}[theorem]{Proposition}
\newtheorem{lemma}[theorem]{Lemma}

\theoremstyle{remark}

\newcommand{\eqa}{\begin{eqnarray}}
\newcommand{\eeqa}{\end{eqnarray}}
\newcommand{\beq}{\begin{equation}}
\newcommand{\eeq}{\end{equation}}
\newcommand{\nn}{\nonumber}

\def\d{\partial}

\def\f{\frac}

\numberwithin{equation}{section}

\begin{document}

\title{On bi-Hamiltonian deformations of exact pencils of hydrodynamic type}
\author{Alessandro Arsie}
\address{Department of Mathematics, University of Toledo, 2801 W. Bancroft Street, Toledo OH 43604, USA}
\email{alessandro.arsie@utoledo.edu}
\author{Paolo Lorenzoni}
\address{Dipartimento di Matematica, Universit{\`a} di Milano - Bicocca,
Via Roberto Cozzi 53, I-20125 Milano, Italy}
\email{paolo.lorenzoni@unimib.it}



\thanks{Partial support from the University of Toledo  URAF grant ``Lax equation, integrability and applications" is gratefully acknowledged}
\begin{abstract}
In this paper we are interested in non trivial bi-Hamiltonian deformations of the Poisson pencil 
 $\omega_{\lambda}=\omega_2+\lambda \omega_1=u\delta'(x-y)+\f{1}{2}u_x\delta(x-y)+\lambda\delta'(x-y)$.
Deformations are generated by a sequence of vector fields $\{X_2, X_3, X_4, \dots\}$,  where each $X_{k}$ is homogenous of degree $k$ with respect to a grading induced by rescaling.  Constructing recursively the vector fields $X_{k}$ one obtains two types of relations involving their unknown coefficients: one set of linear relations and  an other one which involves quadratic relations. We prove that the set of linear relations has a geometric meaning: using Miura-quasitriviality the set of linear relations expresses the tangency of the vector fields $X_{k}$ to the symplectic leaves of $\omega_1$ and this  tangency condition is equivalent to the exactness of the pencil $\omega_{\lambda}$.
Moreover, extending the results of \cite{L}, we construct the non trivial deformations of the Poisson pencil 
 $\omega_{\lambda}$, up to the eighth order in the deformation parameter, showing therefore that deformations are unobstructed and that both Poisson structures are polynomial in the derivatives of $u$  up to that order.
 
\end{abstract}

\maketitle
\section{Introduction}
One of the most relevant problem in the modern theory of integrable systems is the classification of systems of PDEs of evolutionary type of the following form: 
\begin{equation}\label{eq1.eq}
u^i_t=F^i(u, u_x, u_{xx}, \dots, u_{(n)}, \dots)
\end{equation}
where $u^i$ are functions of variables $(x,t)$, $u^i_x$ denotes first derivative of $u^i$ with respect to $x$ and $u^i_{(n)}$ is the $n$-th derivative with respect to $x$ and $i=1, \dots, N$. In particular  the functions $u^i$ are required to satisfy the usual boundary conditions of the Formal Calculus of Variations, i.e. either $u^i\in C^{\infty}(S^1, \mathbb{R})$ or $u^i \in C^{\infty}(\mathbb{R}, \mathbb{R})$ with vanishing conditions at infinity. 
To fix the ideas, in this work we will restrict out attention to the periodic case.  Moreover we assume that the functions $F^i$ are polynomials in the derivatives with respect $x$ of the functions $u^j$.

Systems of the form \eqref{eq1.eq} can be thought as dynamical systems on the space of infinite jets $J^{\infty}(S^1, \mathbb{R}^n)$ considering the complete set of differential consequences of \eqref{eq1.eq}: 
$$(u^i_{(m)})_t=\partial^m_{x} F^i(u, u_x, u_{xx}, \dots, u_{(n)}, \dots),$$
where $\partial_x$ is the total derivative with respect to $x$ acting on $F^j$ as follows:
\begin{equation}\label{eq2.eq}
\partial_x F^j(u, u_x, u_{xx}, \dots, u_{(n)}, \dots):=\sum_{k=0}^{\infty}\sum_{i=1}^N\frac{\partial F^j}{\partial u^i_{(k)}}u^i_{(k+1)}
\end{equation}
and $u^i_{(0)}$ stands for $u^i$. 

Rescaling the independent variables $t\mapsto \epsilon t$, $x \mapsto \epsilon x$, the system \eqref{eq1.eq} transforms to 
\begin{equation}\label{eq3.eq}
u^i_t=\frac{1}{\epsilon}F^i_0(u)+ V^i_j(u) u^j_x+ \sum_{k=1}^\infty \epsilon^k F_k^i(u, u_x, u_{xx}, \dots, u_{(n)}, \dots)
\end{equation}
where $F^i_0$ do not depend on the derivatives of the functions $u$ with respect to $x$, and $F^i_k$ are homogeneous polynomials of degree $k$ in the derivatives of $u$, assigning degree $0$ to functions of $u$ and $\deg(u^i_{(l)})=l$.
Even though it is not strictly necessary, the rescaling of the independent variables and the ensuing presence of the parameter $\epsilon$ have the effect to clearly separate the various homogenous components: this is a key aspect of the perturbative approach to the classification problem. 

In particular in this work we will be concerned with the classification of systems of PDEs of type \eqref{eq3.eq} having a well-defined dispersionless limit  as $\epsilon$ goes to zero. In order to have a well-defined system in this case, it is necessary to restrict our analysis to the family of PDEs in which $F^i_0(u)$ is identically zero. 
Therefore, from now on we restrict our attention to systems of the following form:
\begin{equation}\label{eq4.eq}
u^i_t= V^i_j(u) u^j_x+ \sum_{k=1}^\infty \epsilon^k F_k^i(u, u_x, u_{xx}, \dots, u_{(n)}, \dots).
\end{equation}

The main idea of the perturbative approach to classification is to deal with the terms $F^i_k$, $k\geq 1$ as perturbations of the dispersionless limit $u^i_t=V^i_j(u) u^j_x$ and to reconstruct the integrability of the full system \eqref{eq4.eq} starting from the integrability \cite{T} of the quasilinear system 
\begin{equation}\label{qls}
u^i_t=V^i_j(u) u^j_x.
\end{equation}

A trivial way to produce integrable perturbations starting from a quasilinear system \eqref{qls} consists
 in a change of dependent variables of the form
\begin{equation}
\label{Miuratr}
\tilde{u}^i=F^i_0(u)+\sum_k\epsilon^k F^i_k(u,u_x,u_{xx},\dots)
\end{equation}
where $F_k$ are differential polynomials (in the derivatives of the $u^i$s) of degree $k$ and ${\rm det}\f{\d F_0}{\d u}\ne 0$. Such transformations are called \emph{Miura transformations}. We will not  assume the convergence of the series in the right hand side of \eqref{Miuratr}.  In this formal setting two systems  of the form \eqref{eq4.eq} which are related by a Miura transformation will be considered \emph{equivalent}.

In the perturbative point of view to the classification, different approaches are possible (and have been explored)
 
- Fix a local Hamiltonian structure and extend to all order the conservation laws of the unperturbed system in a recursive way \cite{D}. 

- Extend  to all orders the symmetries of the hydrodynamic limit \cite{S}.

- One approach is based on the additional assumptions  that the systems  \eqref{eq1.eq} one is dealing with  are reductions of a $(2+1)$ integrable PDE \cite{FM}. 

- One other approach is the approach proposed by Dubrovin and Zhang in \cite{DZ1}. It can be applied to
 quasilinear system possessing a local bi-Hamiltonian structure. In this case, instead of studying the deformations of the system it is more convenient to study and classify the deformations of its bi-Hamiltonian structure. 

The aim of this paper is to apply the approach of Dubrovin and Zhang to the simplest possible case:
 the local bi-Hamiltonian structure 
\begin{equation}\label{MPP}
\{u(x),u(y)\}_{\lambda}=2u\delta'(x-y)+u_x\delta(x-y)
-\lambda\delta'(x-y).
\end{equation}
of the Hopf equation
$$u_t=uu_x.$$
The deformations of this structure has been classified,  up to the fourth order, in \cite{L}. 
They depend on a certain number of parameters. All  these parameters, except one, are irrelevant and correspond to Miura equivalent deformations. The remaining one parametrizes the space of non equivalent
 deformations.

The starting observation of the present paper was that, using the freedom in the choice 
 of the irrelevant parameters, one can write the fourth order deformations of the pencil \eqref{MPP}
 in the simple form
\begin{eqnarray*}
&&2u(x)\delta'(x-y)+u_x\delta(x-y)-\lambda\delta'(x-y)\\
&&+\epsilon^2\left\{\d_x\left[c_2(u)\delta^{(2)}(x-y)\right]+\d_x^2\left[c_2(u)\delta'(x-y)\right]\right\}\\
&&+\epsilon^3\left\{-\d_x\left[c_3(u)\delta^{(3)}(x-y)\right]+\d_x^3\left[c_3(u)\delta'(x-y)\right]\right\}+\\
&&+\epsilon^4\left\{\d_x\left[c_4(u)\delta^{(4)}(x-y)\right]+\d_x^4\left[c_4(u)\delta'(x-y)\right]\right\}\\
&&+O(\epsilon^5)
\end{eqnarray*}
where $c_2$ is the relevant functional parameter, $c_3$ is a free parameter, and $c_4=-\f{\d}{\d u}(c_2)^2$. 
 This observation suggested us to look for higher order deformations of the same form:
\begin{eqnarray*}
&&2u(x)\delta'(x-y)+u_x\delta(x-y)-\lambda\delta'(x-y)+\\ 
&&\sum_k\epsilon^k\left\{(-1)^k\d_x\left[c_k(u)\delta^{(k)}(x-y)\right]+\d_x^k\left[c_k(u)\delta'(x-y)\right]\right\}.
\end{eqnarray*}
For a suitable choice of the functions $c_k$ one obtains 
 the two known cases corresponding to KdV and Camassa-Holm hierarchy. In the KdV case
 the function $c_2$ is constant and all the remaining functions vanish, while in the Camassa-Holm 
 case we have $c_k=u$ for all $k\ge 2$, $k$ even, while $c_k=-u$ for $k\geq3$, $k$ odd. So, motivated by this preliminary observation, we started
 to study higher order deformations of \eqref{MPP}. We realized very soon (at the order six) that
 our optimistic conjecture about the form of the deformations was wrong. However we realized also
 that part of the constraints imposed by the Jacobi identity has a simple geometrical interpretation:
 the deformed Poisson pencil can be always written in the form
$${\rm Lie}_{X_{\epsilon}}\delta'(x-y)+O(\epsilon^8)-\lambda\delta'(x-y)$$
where the vector field $X_{\epsilon}$ is always \emph{tangent} to the symplectic leaves of $\delta'(x-y)$.
   Using some important results due to Dubrovin,  Liu and Zhang, we have been able to prove this result to all orders:
   indeed we prove that the tangency of $X_{\epsilon}$ to the symplectic leaves of $\delta'(x-y)$ is valid at any order in the deformation parameter and depends crucially on the \emph{exactness} of the 
 pencil \eqref{MPP} (indeed it is an equivalent condition).
\newline
\newline
The paper is organized as follows. In Section 2 we fix the notations and we review the basic results about
 the subject. This section does not contain new material,
except that we provide a complete proof of Proposition 4, which was instead
sketched in [15]. Due to the amount of results accumulated in the last years we believe that it might be helpful for the reader to have a short review of them. Moreover part of these results will be used in the remaining sections. In Section 3 we recall a powerful formalism developed by several authors which is very convenient to make computations and we state two computational lemmas that are proved in a final Appendix. Although proofs are not difficult, we have not been able to locate them in the literature. 
In Section 4 we use such a formalism to show that the form of the deformations suggested by lower order deformations is unfortunately too optimistic. In Section 5 we prove that the vector field generating the deformation $X_{\epsilon}$ is tangent to the symplectic leaves of $\delta'(x-y)$ at any order (provided the deformation exists) if and only if the undeformed Poisson pencil is exact. Always in Section 4 we find the general form for a scalar pencil to be exact. These results will help to simplify the computations (performed with the help of Maple) of the subsequent Section 6. The corrective terms are computed  in Section 6 where we extend the results of \cite{L} up to the eighth order in the deformation parameter, proving that the deformation is unobstructed up to that order. Section 7 provides some remarks and final comments.

\section{The Dubrovin-Zhang bi-Hamiltonian approach}

A very important class of integrable systems is given by bi-Hamiltonian systems introduced for the first time in \cite{M}.
A system is bi-Hamiltonian if it can be written as a Hamiltonian system with respect to two compatible Poisson structures $\omega_1$ and $\omega_2$, where compatibility conditions entails the fact that $\omega_1+\lambda \omega_2$ is a Poisson structure for any value of $\lambda\in \mathbb{R}$. 
For this class of systems, the presence of a bi-Hamiltonian structure captures all the integrability properties. 

In the case of systems of hydrodynamic type, the class of Hamiltonian structures to be considered was introduced by Dubrovin and Novikov. Let us briefly outline the key points in their construction.
Consider functionals 
$${\mathcal F}[u]:=\int_{S^1} f(u, u_{x}, u_{xx}, \dots) \; dx,$$
and
$$G[u]:=\int_{S^1}g(u, u_{x}, u_{xx}, \dots)\; dx$$
and define a bracket between them as follows:
\begin{equation}\label{eq6.eq}
\{F, G\}[u]:=\iint_{S^1\times S^1}\frac{\delta F}{\delta u^i(x)}\, P^{ij}(x,y)\, \frac{\delta G}{\delta u^j(y)}\;dxdy,
\end{equation}
where $\frac{\delta }{\delta u^i}$ denotes the variational derivative with respect to $u^i$ and the bivector $P^{ij}$ has the following form 
\begin{equation}\label{eq7.eq}
P^{ij}=g^{ij} \delta'(x-y)+\Gamma^{ij}_k u^k_{x} \delta(x-y).
\end{equation}

A deep result of Dubrovin and Novikov characterizes under which conditions the bracket \eqref{eq6.eq} is Poisson: 
\begin{theorem}\label{th1.th}\cite{DN84}
If $\det(g^{ij})\neq 0$, then the bracket \eqref{eq6.eq} is Poisson if and only if the metric $g^{ij}$ is flat and the functions $\Gamma^{ij}_k$ are related to the Christoffel symbols of $g_{ij}$ (the inverse of $g^{ij}$) by the formula $\Gamma^{ij}_k=-g^{il} \Gamma^j_{lk}$.
\end{theorem}

In this work, when referring to brackets of Dubrovin-Novikov in the previous theorem, we will call them brackets of hydrodynamic type, discarding the case of non-local brackets of hydrodynamic type arising from non-flat metrics. 
Let us remark that in the flat coordinates $\{f_1, \dots, f_N\}$, the brackets of form \eqref{eq6.eq} reduce to 
\begin{equation}\label{eq7bis.eq} 
\eta^{ij} \delta'(x-y)
\end{equation}
where $\eta^{ij}$ is a constant matrix. Observe also that such a bracket is clearly degenerate and its Casimirs are the integrals of the flat coordinates:  
$$C_j=\int_{S^1} f_j \;dx.$$

In this set-up we have the following definition:
\begin{definition}\label{def1.def}
 A bi-Hamiltonian structure of hydrodynamic type is given by a pair of Poisson bivectors $P^{ij}_1$, $P^{ij}_2$ satisfying separately the conditions of Theorem \ref{th1.th} (for two different contravariant metrics $g^{ij}_1$, $g^{ij}_2$) and satisfying an additional compatibility condition requiring $Q^{ij}_{\lambda}:=P^{ij}_1+\lambda P^{ij}_2$ to be Poisson for any $\lambda$. 
Dubrovin proved that from a differential-geometric point of view the compatibility of $P_1$ and $P_2$ is equivalent to the fact that $g_{\lambda}:=g^{ij}_1+\lambda g^{ij}_2$ is a flat pencil of metrics which means \cite{D07}: 
\begin{enumerate}
\item the Riemann tensor $R_{\lambda}$ of the pencil $g_{\lambda}$ vanishes for any value of $\lambda$;
\item the Christoffel symbols $(\Gamma_{\lambda})^{ij}_{k}$ of the pencil are given by $(\Gamma_1)^{ij}_{k}+\lambda (\Gamma_2)^{ij}_{k}$.
\end{enumerate}
\end{definition}

Since the bi-Hamiltonian structure encodes all the characteristics of an integrable system, Dubrovin and Zhang proposed to study the integrable perturbations of systems of PDEs of the type described above, by studying perturbations of their associated bi-Hamiltonian structure and classifying them modulo Miura transformations. 
 They also conjectured that, under suitable additional assumptions coming from the Gromov-Witten theory,
 the perturbed bi-Hamiltonian hierarchies exist and are uniquely determined by their dispersionless limit. They made also important steps towards the proof of the conjecture. One of the very important missing gaps, concerning the polynomiality of the dispersive corrections was recently filled in the remarkable preprint  \cite{BPS}, where it is proved also that one of the Hamiltonian structure is polynomial. However the polynomiality of the second Hamiltonian structure to all orders in the deformation parameter $\epsilon$ is still an open problem.
 
\subsection{Poisson structures on the space of loops}
On the space of smooth loops $\mathcal{L}(\mathbb{R}^n):=\{h: S^1 \rightarrow \mathbb{R}^n, h\in C^{\infty}\}$ we consider the ring $\mathcal{A}$ of differential polynomials: 
\begin{equation}\label{eq8.eq}
f(x, u, u_x, \dots):=\sum_{i_1, s_1, \dots, i_m, s_m}f_{i_1, s_1; \dots; i_m, s_m}(x, u)u^{i_1}_{(s_1)}\dots u^{i_m}_{(s_m)},
\end{equation}
where $u=(u^1, \dots, u^N)$, $u_{(s)}=(u^1_{(s)}, \dots, u^N_{(s)})$ with $u^i_{(s)}:=\frac{d^s}{dx^s}u^i(x)$. Moreover the coefficients $f_{i_1, s_1; \dots ;i_m, s_m}(x, u)$ of these differential polynomial are required to be smooth functions on $S^1\times \mathbb{R}^n$. 

Denote by $\mathcal{A}_0:=\mathcal{A}/\mathbb{R}$ the space of differential polynomials modulo constants, and 
$\mathcal{A}_1:=\mathcal{A}_0 \; dx$. Then one has a well-defined map $d: \mathcal{A}_0\rightarrow \mathcal{A}_1$:
\begin{equation}\label{eq9.eq}
f\mapsto df:=\left(\frac{\partial f}{\partial x}+\sum_{i, s} \frac{\partial f}{\partial u^i_{(s)}} u^i_{(s+1)} \right)\; dx.\end{equation}
The quotient $\Lambda:=\mathcal{A}_1 /d\mathcal{A}_0$ is called the space of local functionals on $\mathcal{L}(\mathbb{R}^n)$. Elements of $\Lambda$ (local functionals) are expressed as integrals over $S^1$ of a representative differential polynomial: 
\begin{equation}\label{eq10.eq}\Lambda \ni \lambda =\int_{S^1} f(x, u, u_x, \dots, u_{(n)}) \; dx.\end{equation}
Observe that by the fact that we are dealing with suitable boundary conditions, two elements $\lambda_1$ and $\lambda_2$ are identified by a differential polynomial up to a total derivative. 

In order to study Poisson bi-vectors on the loop space, we need to introduce the notion of local multi-vectors (k-vectors). 
A local k-vector $\alpha$ on the loop space $\mathcal{L}(\mathbb{R}^n)$ is defined to be a formal possibly infinite sum 
\begin{equation}\label{eq11.eq}\alpha=\sum \frac{1}{k!} \partial^{s_1}_{x_1}\dots \partial^{s_k}_{x_k}A^{i_1, \dots, i_k} \frac{\partial }{\partial u^{i_1}_{(s_1)}(x_1)}\wedge \dots \wedge \frac{\partial }{\partial u^{i_k}_{(s_k)}(x_k)},\end{equation}
with coefficients 
\begin{equation}\label{eq12.eq}
A^{i_1, \dots, i_k}=\sum_{p_{2}, \dots, p_{k}\geq 0}B^{i_1, \dots, i_k}_{p_{2}, \dots, p_{k}}(u(x_1), u_x(x_1), \dots)\delta^{(p_{2})}(x_1-x_2)\dots, \delta^{(p_{k})}(x_1-x_k).\end{equation}
The coefficients $B^{i_1, \dots, i_k}_{p_{2}, \dots, p_{k}}(u(x_1), u_x(x_1), \dots)$ belong to $\mathcal{A}$, the ring of differential polynomials. The coefficients $A^{i_1, \dots, i_k}$ are skew-symmetric with respect to simultaneous exchange $(i_r, x_r)$ with $(i_s, x_s)$ and they are called the components of the $\alpha$ k-vector. 

The space of local k-vectors is denoted as $\Lambda^k_{\text{loc}}$.
Specializing to the case of $1$-vectors, we obtain the class of local vector fields on the loop space $\mathcal{L}(\mathbb{R}^n)$. They are expressed by the following formula: 
\begin{equation}\label{eq15.eq}\xi=\sum_{i=1}^N \sum_{s\geq 0}\partial^s_x \xi^i(u(x), u_x(x), \dots) \frac{\partial}{\partial u^i_{(s)}}.\end{equation}
Their components do not depend explicitly on $x$ and thus they are called translation invariant evolutionary vector fields. 

The subspace $\Lambda^0_{\text{loc}}$ of $\Lambda_{\text{loc}}$ is identified with the space of local functionals of the form 
\begin{equation}\label{eq16.eq}
F:=\int_{S^1}f(u(x), u_x(x), \dots) \; dx, \quad f(u(x), u_x(x), \dots)\in \mathcal{A}_0.
\end{equation}

Since we will focus our attention to local bivectors, we provide their general expression. A local bivector $\omega$ has the form 
\begin{equation}\label{eq17.eq}
\omega=\frac{1}{2}\sum \partial^r_x \partial^s_y \omega^{ij} \frac{\partial }{\partial u^i_{(r)}(x)}\wedge\frac{\partial }{\partial u^j_{(s)}(y)},
\end{equation}
where 
\begin{equation}\label{eq18.eq}
\omega^{ij}=A^{ij}(x-y; u(x), u_x(x), \dots)=\sum_{t\geq 0}A^{ij}_t(u(x), u_x(x), \dots) \delta^{(t)}(x-y).
\end{equation}

In order to characterize which local bivectors corresponds to a Poisson bivector, it is necessary to introduce on the space of local multi-vectors with its natural gradation 
$$\Lambda^*_{\text{loc}}:=\Lambda^0_{\text{loc}}\oplus \Lambda^1_{\text{loc}}\oplus \Lambda^2_{\text{loc}} \oplus \dots,$$
a bilinear operation:
\begin{equation}
[\cdot, \cdot]: \Lambda^r_{\text{loc}}\times \Lambda^s_{\text{loc}}\rightarrow \Lambda^{r+s-1}_{\text{loc}}, \quad r, s\geq 0,
\end{equation} 
called Schouten-Nijenhuis bracket. 
Let's describe how the Schouten-Nijenhuis bracket operates on certain pairs of local multi-vectors. We have that for any $F, G\in \Lambda^{0}_{\text{loc}}$ $[F, G]=0$ identically, while if $\xi\in \Lambda^1_{\text{loc}}$ is of the form \eqref{eq15.eq} and $F$ is local functional of the form \eqref{eq16.eq}, then
\begin{equation}\label{eq19.eq}
[\xi, F]=\int_{S^1}\sum_{t\geq 0}\sum_{i=1}^N (\partial^t_x \xi^i) \frac{\partial f}{\partial u^i_{(t)}} \; dx=\int_{S^1}\sum_{i=1}^N \xi^i \frac{\delta F}{\delta u^i(x)} \; dx
\end{equation}
where $$\frac{\delta F}{\delta u^i(x)}=\sum_{t\geq 0}(-1)^t \partial^t_x \left( \frac{\partial f}{\partial u^i_{(t)}}\right)$$
is the variational derivative of the local functional $F$. Observe that $[\xi, F]$ is indeed an element of $\Lambda^0_{\text{loc}}$. 
The Schouten-Nijenhuis bracket of two local vector fields $\xi, \eta$ is again a vector field $\mu$ described by the following formula
\begin{equation}\label{eq20.eq}
\mu =[\xi, \eta]=\sum_{s,i} \partial^s_x \mu^i \frac{\partial }{\partial u^i_{(s)}}=\sum_{s,i,j,t}\partial^s_x\left(\xi^j_{(t)}\frac{\partial \eta^i}{\partial u^j_{(t)}}- \eta^j_{(t)}\frac{\partial \xi^i}{\partial u^j_{(t)}}\right) \frac{\partial }{\partial u^i_{(s)}}.
\end{equation}
where the components $\mu^i$ of $\mu$ are given as
$$\mu^i=\sum_{j,t} \xi^j_{(t)}\frac{\partial \eta^i}{\partial u^j_{(t)}}- \eta^j_{(t)}\frac{\partial \xi^i}{\partial u^j_{(t)}}$$ 

The Schouten-Nijenhuis bracket of a local bivector $\omega$ of the form \eqref{eq17.eq} and a local functional $F$ gives rise to a local vector field whose components are
\begin{equation}\label{eq21.eq}
[\omega, F]^i=\sum_{j,k}A^{ij}_k \partial^k_x \frac{\delta F}{\delta u^j(x)}.
\end{equation}
Analogously the Schouten-Nijenhuis bracket of a local bivector $\omega$ and local vector field $\xi$ is again a local bivector whose components are given by 
\begin{eqnarray}\label{eq22.eq}
&&\quad[\omega, \xi]^{ij}=\\
&&\sum_{k,s}\left( \partial^s_x \xi^k(u(x), \dots) \frac{\partial A^{ij}}{\partial u^k_{(s)}(x)}- \frac{\partial \xi^i(u(x), \dots)}{\partial u^k_{(s)}(x)} \partial ^s_xA^{kj}-\frac{\partial \xi^j(u(y), \dots)}{\partial u^k_{(s)}(y)}\partial^s_y A^{ik}\right).\nn
\end{eqnarray}

Finally, if $P$ and $Q$ are two translation invariant bivectors, then their Schouten - Nijenhuis bracket $[P,Q]$ is a translation invariant trivector, whose complicated formula can be found in \cite{DZ1} together with many other details.


\noindent{\bf Remark 1}
{\it An alternative efficient way to compute  Schouten brackets is based on the idea of substituting multivectors with superfunctional
 and Schouten bracket with Poisson brackets between superfunctionals \cite{G}, see also \cite{B}. We recall this formalism in Section 3. 
}

The Schouten-Nijenhuis bracket satisfies also the following properties (a graded Jacobi identity and a graded skew-symmetry):
\begin{equation}\label{eq23.eq}
(-1)^{ki}[[a,b] c]+(-1)^{kl}[[b,c],a]+(-1)^{li}[[c,a],b]=0, \text{ graded Jacobi identity}
\end{equation} 
\begin{equation}\label{eq24.eq}
[a,b]=(-1)^{kl}[b,a], a\in \Lambda^k_{\text{loc}}, b\in \Lambda^l_{\text{loc}}, c\in \Lambda^i_{\text{loc}}.
\end{equation}

The importance of the Schouten-Nijenhuis bracket stems from the fact that a local Poisson structure can be characterized in the following way: 
\begin{definition}\label{def2.def}
A local bivector $\omega\in \Lambda^2_{\text{loc}}$ of the form \eqref{eq17.eq} is a local Poisson structure on $\mathcal{L}(\mathbb{R}^n)$ if $[\omega, \omega]=0$. 
\end{definition} 
A local Poisson structure gives rise to a Poisson bracket on the space of local functionals in the following way: 
\begin{equation}\label{eq25.eq}
\{ F, G\}:=\int_{S^1} \sum_{k\geq 0}\frac{\delta F}{\delta u^i(x)}A^{ij}_k(u, u_x, \dots) \partial^k_x \frac{\delta G}{\delta u^j(x)}\; dx.
\end{equation}
Using a special choice of local functionals $F=\int u^i(w) \delta(w-x) \;dw$, $G=\int u^j(w) \delta(w-y)\; dw$, we recover the usual representation of a Poisson structure as 
$$\{ u^i(x), u^j(x)\}=\sum_{k\geq 0} A^{ij}_k(u(x), u_x(x), \dots)\delta^{(k)}(x-y). $$

As we have seen in the introduction (see equation \eqref{eq3.eq}), the scaling $\psi_{\epsilon}:x\mapsto x\epsilon$ decomposes the evolutionary vector field into homogenous components. Analogously, this same scaling induces a natural gradation on the space $\Lambda^k_{\text{loc}}$. Now we detail how the various ingredients rescale under $\psi_{\epsilon}$. First of all we define $(\psi_{\epsilon} u^i)(x):=u^i(\epsilon x)$.
It is immediate to see that $$[(\psi_{\epsilon} u^i)(x)]_{(s)}=\epsilon^s u^i_{(s)}$$
and consequently
$$\frac{\partial}{\partial u^i_{(s)}(\epsilon x)}=\frac{1}{\epsilon^s}\frac{\partial}{\partial u^i_{(s)}}.$$
Moreover, to see how the $\delta$ distribution and its derivatives rescale, consider 
$$\int f(x) \delta^{(s)}(\epsilon x) \; dx=\int f\left(\frac{z}{\epsilon}\right) \delta^{(s)}(z) \; \frac{dz}{\epsilon}=(-1)^s \int
 \frac{d^s f \left(\frac{z}{\epsilon}\right)}{dz^s} \delta(z)  \; \frac{dz}{\epsilon}=$$
 $$=(-1)^s f_{(s)}(0)\frac{1}{\epsilon^{s+1}}=\int f(x) \delta^{(s)}(x)\frac{1}{\epsilon^{s+1}} \; dx, $$
 from which 
 $$  \delta^{(s)}(\epsilon x)=\delta^{(s)}(x)\frac{1}{\epsilon^{s+1}},$$
 or 
 $$\delta^{(s)}(x)=\epsilon^{s+1}(\psi_{\epsilon}(\delta^{(s)}))(x).$$
 With these information, we can show that the rescaling $\psi_{\epsilon}$ induces a decomposition on $\Lambda^{k}_{\text{loc}}$ into monomials of different degrees. For simplicity we focus on the cases of $\Lambda^{1}_{\text{loc}}$ and $\Lambda^2_{\text{loc}}$. 
Any local vector field has the form 
\begin{equation}\label{eq26.eq}
\xi=\sum_{i=1}^N \sum_{s\geq 0}\partial^s_x \xi^i(u(x), u_x(x), \dots) \frac{\partial}{\partial u^i_{(s)}}, 
\end{equation}
where its components $\xi^i$ are elements of the ring $\mathcal{A}$ of differential polynomials. Since the rescalings induced in $\partial^s_x$ and $ \frac{\partial}{\partial u^i_{(s)}}$ are one the reciprocal of the other, the splitting of the $\xi$ into homogeneous monomials depends only on its components $\xi^i$. 
In general the components $\xi^i$ split under rescaling into homogeneous monomials as follows: 
\begin{equation}\label{eq27.eq}
\xi^i:=a^i(u)+\epsilon\sum_{j=1}^N b^i_j(u) u^j_x +\epsilon^2\left(\sum_{j=1}^N e^i_{j}(u) u^j_{xx}+\sum_{j,l=1}^N h^i_{jl}(u)u^{j}_x u^{l}_x\right)+\dots,
\end{equation}
and this gives rise to the an analogous decomposition as $$\Lambda^1_{\text{loc}}=\bigoplus_{k=0}^{\infty}\Lambda^1_{k, \text{loc}},$$
where $\Lambda^1_{k, \text{loc}}$ is the space of local vector fields $\xi$ whose components $\xi^i$ are 
homogeneous differential polynomials of degree $k$. 
Let's apply the same analysis to the case of elements of $\Lambda^2_{\text{loc}}$. 
Recall that a local bivector $\omega$ is given by 
$$\omega=\frac{1}{2}\sum \partial^r_x \partial^s_y \omega^{ij} \frac{\partial }{\partial u^i_{(r)}(x)}\wedge\frac{\partial }{\partial u^j_{(s)}(y)},$$
where 
$$
\omega^{ij}=A^{ij}(x-y; u(x), u_x(x), \dots)=\sum_{t\geq 0}A^{ij}_t(u(x), u_x(x), \dots) \delta^{(t)}(x-y).$$
Since the terms $ \partial^r_x \partial^s_y$ and $\frac{\partial }{\partial u^i_{(r)}(x)}\wedge\frac{\partial }{\partial u^j_{(s)}(y)}$ have reciprocal scaling factors, the decomposition of $\omega$ in homogeneous monomials under rescaling is completely controlled by the way in which its components $\omega^{ij}$ decompose. 
Rewrite $\omega^{ij}$ as follows:
\begin{equation}\label{eq30.eq}
\sum_{t\geq 0} \sum_{l\geq 0} (A^{ij}_t)_l  \delta^{(t)}(x-y),
\end{equation}
where $(A^{ij}_t)_l$ is the homogenous components of degree $l$ of the differential polynomial $A^{ij}_t$. 
Rescaling under $\psi_{\epsilon}$ gives rise to
$$\sum_{t\geq 0} \sum_{l\geq 0} (A^{ij}_t)_l \epsilon^{l+t+1}  \delta^{(t)}(x-y), $$
which can be rewritten setting $k=l+t+1$ as
\begin{equation}\label{eq31.eq}
\sum_{k=1}^{\infty}\epsilon^k\sum_{t=0}^{k-1} (A^{ij}_t)_{k-1-t} \delta^{(t)}(x-y).
\end{equation}
In this way the components of $\omega^{ij}$ of a bivector decompose in homogeneous terms $[\omega^{ij}]_k$ of the form 
\begin{equation}\label{eq31bis.eq}
[\omega^{ij}]_k:= \sum_{t=0}^{k-1} (A^{ij}_t)_{k-1-t} \delta^{(t)}(x-y).
\end{equation}
Therefore we have an induced decomposition of $\Lambda^2_{\text{loc}}$ as follows: 
$$\Lambda^2_{\text{loc}}=\bigoplus_{k\geq 1} \Lambda^{2}_{k, \text{loc}},$$
where $\omega$ is in $\Lambda^{2}_{k, \text{loc}}$ exactly when its components $\omega^{ij}$ are of the form
of the addends in \eqref{eq31.eq}. 

Notice that the bivectors of hydrodynamic type $\omega$ introduced by Dubrovin and Novikov (\eqref{eq7.eq}) are elements of $\Lambda^{2}_{2, \text{loc}}$, and indeed any element of $\Lambda^{2}_{2, \text{loc}}$ is a bivector of hydrodynamic type. 
Following \cite{DN84}, we call any Poisson structure of the form \begin{equation}\label{eq32.eq}
(\omega+P)\in \Lambda^2_{\text{loc}},\end{equation} where 
$P=\sum_{k\geq 1} P_k$, $P_k\in \Lambda^2_{k+2,\text{loc}}$ a deformation of $\omega$. Notice that due to the rescaling $\psi_{\epsilon}$ the deformation \eqref{eq32.eq} transforms to $\epsilon^2(\omega+\sum_{k\geq 1} \epsilon^k P_k)$, so we can re-write the deformation as 
$$\omega+\sum_{k\geq 1} \epsilon^k P_k.$$
As we did for bivectors, the space of $j$-multivectors $\Lambda^{j}_{\text{loc}}$ can be decomposed in terms that are homogenous under rescaling: 
$$\Lambda^{j}_{\text{loc}}=\oplus_k \Lambda^j_{k,\text{loc}}.$$
Any element $P\in \Lambda^j_{k,\text{loc}}$ is transformed to $\epsilon^k P$ under rescaling. 

\subsection{Poisson cohomology}
Fix $\omega\in \Lambda^{2}_{2, \text{loc}}$ and assume that $\omega$ is Poisson.  Consider the map 
$$d_{\omega}: \Lambda^j_{\text{loc}}\rightarrow \Lambda^{j+1}_{\text{loc}}, \quad d_{\omega}(a)=[\omega, a].$$
In particular, it is immediate to see that the map $d_{\omega}$ maps $\Lambda^{j}_{k, \text{loc}}$ to $\Lambda^{j+1}_{k+2, \text{loc}}$.  
This map has the property that $d^2_{\omega}=0$ identically, due to the graded Jacobi identity satisfied by the Schouten-Nijenhuis bracket and the fact that $\omega$ is Poisson. 
This enables one to define cohomology groups, known as Poisson cohomology groups in the following way: 
\begin{equation}\label{eq33bis.eq}
H^j(\mathcal{L}(\mathbb{R}^n), \omega):=\frac{\ker\{d_{\omega}: \Lambda^j_{\text{loc}}\rightarrow \Lambda^{j+1}_{\text{loc}}\}}
{\mathrm{im}\{d_{\omega}: \Lambda^{j-1}_{\text{loc}} \rightarrow \Lambda^j_{\text{loc}}\}}
\end{equation}
These cohomology groups are completely analogous to those defined in the case of finite dimensional Poisson manifolds by 
Lichnerowicz \cite{lichn}. 

Due to the fact that the space of $j$-multivectors $\Lambda^{j}$ has a natural decomposition in terms of components homogenous under rescaling and the fact that $d_{\omega}$ preserves this homogenous decomposition, each cohomology group inherits a natural decomposition in homogeneous parts.
Indeed, we can introduce
\begin{equation}\label{eq33.eq}
H^j_k(\mathcal{L}(\mathbb{R}^n), \omega):=\frac{\ker\{d_{\omega}: \Lambda^j_{k,\text{loc}}\rightarrow \Lambda^{j+1}_{k+2,\text{loc}}\}}
{\mathrm{im}\{d_{\omega}: \Lambda^{j-1}_{k-2,\text{loc}} \rightarrow \Lambda^j_{k,\text{loc}}\}}
\end{equation} 
where a class $[\alpha]$ is in $H^j_k(\mathcal{L}(\mathbb{R}^n), \omega)$ exactly when any of its representatives can be chosen in $\Lambda^j_{k,\text{loc}}$. 
Naturally, one has 
\begin{equation}\label{eq34.eq}
H^j(\mathcal{L}(\mathbb{R}^n), \omega)=\oplus_{k} H^j_k(\mathcal{L}(\mathbb{R}^n), \omega).
\end{equation}
Let us remark that a decomposition like \eqref{eq34.eq} is typical of the infinite dimensional situation, and it does not have an analogous correspondence in the finite dimensional case. 

For Poisson structures of hydrodynamic type like \eqref{eq7.eq}, it has been proved in \cite{G} (see also \cite{DMS} for an independent proof
 of the cases $n=1,2$) that $H^k(\mathcal{L}(\mathbb{R}^n), \omega)=0$ for $k=1,2,\dots$. 
  
 The vanishing of these cohomology groups implies that any deformation of $\omega$ of the form 
\begin{equation}\label{DPB}
P^{\epsilon} =\omega +\sum_{n=1}^{\infty}\epsilon^{n}P_{n},
\end{equation}
where $P_k\in \Lambda^2_{k+2, \text{loc}}$ can be obtained from $\omega$ by performing a Miura transformation. 
 Indeed, from the Poisson condition
$$[P^{\epsilon},P^{\epsilon}]=0$$
it follows that $P_1$ is a cocycle of $\omega$ and therefore a coboundary
$$P_1={\rm Lie}_{X_1}\omega,$$
for a suitable vector field $X_1$. This means that, performing the Miura transformation generated by the vector field $-X_1$, we can eliminate the term $\epsilon$ and obtain a local Poisson bivector of the form
$$\tilde{P}=\omega +\sum_{n=2}^{\infty}\epsilon^{n}\tilde{P}_{n}.$$
Using the same arguments we can show that 
$$\tilde{P}_2={\rm Lie}_{X_2}\omega$$ 
and therefore it can be eliminated by the Miura transformation  generated by the vector field $-X_2$. In this way, step by step, we reduce $P$ to $\omega$. The reducing Miura transformation is the composition of the infinite sequence of Miura transformations generated by $-X_1,-X_2,\dots$.

Totally different is the case in which we deform  a pencil of local  bivectors.
 Without loss of generality we can  assume such a pencil of the form
$$P^{\epsilon}_{\lambda}
=\omega_2-\lambda\omega_1+\sum_{n=1}^{\infty}\epsilon^{n}P_{n},$$
where $\omega_1$ and $\omega_2$ are a pair of compatible bivectors of hydrodynamic type:
\begin{equation}\label{eq35.eq}
\omega_{a}=g^{ij}_a\delta'(x-y)+\Gamma^{ij}_{k,a}u^k_x\delta(x-y), \quad a=1,2,
\end{equation}
 Indeed, due to the triviality of $H^2(\mathcal{L}(\mathbb{R}^n), \omega)$, the  $\epsilon$-corrections to $\omega_1$ can be eliminated by a Miura transformation. However, the requirement that $P^{\epsilon}_{\lambda}$ is a Poisson bivector, namely
$$[P^{\epsilon}_{\lambda},P^{\epsilon}_{\lambda}]=0,$$
for any $\lambda\in \mathbb{R}$
imposes some restrictions on the bivector
$P_2^{\epsilon}=\omega_2+\sum_{n=1}^{\infty}\epsilon^{n}P_{n}$

\noindent
First:  it must be compatible with $\omega_1$, that is to say, indicating with $d_1(\cdot):=[\omega_1, \cdot]$
\begin{eqnarray}
\label{rr}d_1 P_{n}&=&0,
\end{eqnarray}
which means that all the terms $P_n$ must be coboundary of $\omega_1$.

\noindent
Second: it must be a Poisson bivector. This is equivalent to the system of conditions
\begin{eqnarray}\label{rec2}
d_2 P_{n}&=&-\frac{1}{2}\sum_{k=1}^{n-1}[P_{k},P_{n-k}],\,\,\,n=1,2,\dots
\end{eqnarray}
\begin{proposition}
The system \eqref{rec2} is compatible.
\end{proposition}
\proof Indeed due to the vanishing of  $H^3(\mathcal{L}(\mathbb{R}^n),\omega_2)$
  compatibility is equivalent to the requirement that the right-hand sides are cocycles of $\omega_2$:
\begin{equation*}
d_2\left(\sum_{k=1}^{n-1}[P_{k},P_{n-k}]\right)=0,\,\,\,n=1,2,\dots.
\end{equation*}
This can be proved by induction using the graded Jacobi identity
\begin{eqnarray*}
&&d_2\left(\sum_{k=1}^{n-1}[P_{k},P_{n-k}]\right)=\\
&&\sum_{k=1}^{n-1}[\omega_2,[P_{k},P_{n-k}]]=\\
&&-\sum_{k=1}^{n-1}[P_{k},[\omega_2,P_{n-k}]]-\sum_{k=1}^{n-1}[P_{n-k},[P_{k},\omega_2]]=\\
&&-2\sum_{k=1}^{n-1}[P_{k},[\omega_2,P_{n-k}]]=\\
&&-2\sum_{k=1}^{n-1}\sum_{l=1}^{n-k-1}[P_{k},[P_{n-k-l},P_{l}]]=\\
&&-2\sum_{k+l+m=n,1\le k\le l\le m}\left([P_{k},[P_{l},P_{m}]]+[P_{m},[P_{k},P_{l}]]+[P_{l},[P_{m},P_{k}]]\right)=0.
\end{eqnarray*}
This means that  we can solve recursively the equations \eqref{rec2} and the compatibility is proved. 
\endproof
 At each step the solution of \eqref{rec2} is defined up to a coboundary
 of $\omega_2$. The problem is to prove that it is always possible to choose these coboundaries in such a way that the resulting
 Poisson bivector is compatible with $\omega_1$. In other words the problem is to prove that any solution $P_n$ of \eqref{rec2}
 has the form 
\begin{equation}\label{bihc}
P_n=d_1  X_n+d_2 Y_n
\end{equation}
if  $P_1,\dots,P_{n-2},P_{n-1}$  are coboundaries of $\omega_1$. This is a non trivial open problem.  Notice that 
 if  $P_1,\dots,P_{n-2},P_{n-1}$  are coboundaries of $\omega_1$, then the bivectors $P_n$ defined by \eqref{rec2} satisfy the condition
\begin{equation}\label{d1d2}
d_1 d_2 P_n=0
\end{equation}
Indeed
\begin{eqnarray*}
&&d_1\left(\sum_{k=1}^{2n-1}[P_{2k},P_{2n-2k}]\right)=\\
&&\sum_{k=1}^{2n-1}[\omega_1,[P_{2k},P_{2n-2k}]]=\\
&&-\sum_{k=1}^{2n-1}[P_{2k},[\omega_1,P_{2n-2k}]]-
\sum_{k=1}^{2n-1}[P_{2n-2k},[P_{2k},\omega_1]]=0
\end{eqnarray*}
Unfortunately this is  not sufficient to conclude that $P_n$ has the form \eqref{bihc}. The possible obstruction lives in the bi-Hamiltonian  cohomology group
$$H^3(\mathcal{L}(\mathbb{R}^n),\omega_1,\omega_2)=\f{{\rm Ker}\left(d_1 d_2\,|_{\Lambda_{\text{loc}}^2}\right)}{{\rm Im}\left(d_1|_{\Lambda_{\text{loc}}^1}\right)\oplus{\rm Im}\left(d_1|_{\Lambda_{\text{loc}}^1}\right)}.$$
To the best of our knowledge, bi-Hamiltonian cohomology groups were initially introduced and studied in \cite{GZ}. However bi-Hamiltonian cohomology groups in the framework of integrable PDEs were first used in \cite{DZ1}.
\newline
Since the deformations $P_{\epsilon}=\sum_{k=1}^{\infty}\epsilon^k P_k$ are coboundaries, namely $P_k={\rm Lie}_{X_k}\omega_1,$
in order to construct explicitly the components $P_k$ it is convenient to solve the equations for the vector fields $X_k$ generating the deformation, instead of solving
 the corresponding equations for the bivectors $P_k$ that, in general, are more involved.
  
Let us consider, for instance, first order defomations, that is  $P^{\epsilon}_{\lambda}:=\omega_2+\epsilon {\rm Lie}_{X_1}\omega_1 -\lambda \omega_1.$
Since we want $P^{\epsilon}_{\lambda}$ to be Poisson up to the order $\epsilon$ included, we require $[P^{\epsilon}_{\lambda}, P^{\epsilon}_{\lambda}]=o(\epsilon)$. This implies 
\begin{equation} 
[\omega_2,P_1]=d_{2}P_1=d_{2}d_{1}X_1
=-d_{1}d_{2}X_1=0 
\end{equation}

\noindent
where we have used the fact that $(d_{1}+d_{2})^{2}=0$. 

\noindent
Among all vector fields that satisfy the equation 
$d_{1}d_{2}X_1=0$ we have to single out those defining trivial deformations, that is those generating 
deformations $P_1$ that can be obtained by infinitesimal change of 
coordinates: 

\begin{eqnarray}
&&{\rm Lie}_{\tilde{X}}\omega_{1}=0\\
&&{\rm Lie}_{\tilde{X}}\omega_{2}=P_1
\end{eqnarray}      
Notice that the vector field $\tilde{X}$ does not coincide with the the vector field $X_{1}$ defining the deformation.

\begin{theorem} 
\label{trivial}
Non trivial first order deformations are the elements of the group  
$$H^2_{2}(\mathcal{L}(\mathbb{R}^n),\omega_1,\omega_2)=\f{{\rm Ker}\left(d_1 d_2\,|_{\Lambda^1_{2,\text{loc}}}\right)}{{\rm Im}\left(d_1|_{\Lambda^0_{0,\text{loc}}}\right)\oplus{\rm Im}\left(d_1|_{\Lambda^0_{0,\text{loc}}}\right)}.$$
\end{theorem} 

\begin{proof}
Suppose that $X_{1}=d_{1}a+d_{2}b$ then
$$P_1=d_{1}d_{2}b=-d_{2}d_{1}b$$
This means $P_{1}=Lie_{\tilde{X}}\omega_{2}$ with
$\tilde{X}=-d_{1}b$.
Moreover 
$$Lie_{\tilde{X}}\omega_{1}=0$$
and therefore the deformation is trivial.

\noindent
Assume now that the deformation is trivial. Then, by definition, we have
$$Lie_{\tilde{X}}\omega_{1}=0,\qquad Lie_{\tilde{X}}\omega_{2}=P_1$$
which implies, due to the vanishing of the first cohomology group,
$$\tilde{X}=d_{1}b.$$
The above condition entails
$$-d_{1}d_{2}b=d_{2}d_{1}b=Lie_{\tilde{X}}\omega_2=
P_1=d_{1}Y_1$$ 
with $Y_1=-d_{2}b+d_{1}a$.
\end{proof}

Higher order deformations can be treated in a similar way. It turns out that non trivial deformations
 are related to the cohomology groups
$$H^2_{k}(\mathcal{L}(\mathbb{R}^n),\omega_1,\omega_2)=\f{{\rm Ker}\left(d_1 d_2\,|_{\Lambda^1_{k,\text{loc}}}\right)}{{\rm Im}\left(d_1|_{\Lambda^0_{k-2,\text{loc}}}\right)\oplus{\rm Im}\left(d_1|_{\Lambda^0_{k-2,\text{loc}}}\right)}.$$ 
The study of such cohomology groups has been done by Liu and Zhang in \cite{LZ} in the semisimple case, that is assuming that the eigenvalues $u^1,\dots,u^n$ of the matrix $g_1^{-1}g_2$ define a set of local coordinates, called
 \emph{canonical coordinates} (here $g_1$ and $g_2$ are the contravariant metrics defining the two undeformed Poisson structures of hydrodynamic type $\omega_1$ and $\omega_2$). Under this additional assumption they showed that 
$$H^2_{k}(\mathcal{L}(\mathbb{R}^n),\omega_1,\omega_2)=0
\quad \forall k\ne 3$$ 
and that the elements of
$$H^2_{3}(\mathcal{L}(\mathbb{R}^n),\omega_1,\omega_2)$$
are vector fields of the form
\begin{equation}\label{qtvf}
d_2\left(\sum_{i=1}^n\int c^i(u^i)u^i_x{\rm log}u^i_x\,dx\right)-d_1\left(\sum_{i=1}^n\int u^i c^i(u^i)u^i_x{\rm log}u^i_x\,dx\right)
\end{equation}
where $c^i(u^i)$ are arbitrary functions of a single variable. 
Notice that the functionals in the brackets in the formula (\ref{qtvf}) do not belong to $\Lambda^0_{1,{\rm loc}}$
 due to the non-polynomial dependence on the $x$-derivatives of the $u$'s. If we allow such a dependence
 all the elements in $H^2_{3}(\mathcal{L}(\mathbb{R}^n),\omega_1,\omega_2)$ become ``trivial".
 This remark justifies the following definitions \cite{DZ1}.

\begin{definition} 
The transformations of the form
$$u^i \to v^i = u^i +\sum_{k=1}^{\infty}
\epsilon^k F^i_k(u; u_x, . . . , u_{(n_k)}),\qquad i = 1, . . . , n$$ 
where the coefficients $F^i_k$ are quasihomogeneous of the degree $k$ rational functions in the
derivatives $u_x, . . . , u_{(n_k)}$ are called quasi-Miura transformations. 
\end{definition}

\begin{definition}
A deformation
 of a Poisson pencil of hydrodynamic type 
 is called quasitrivial if there exists a quasi-Miura transformation
reducing the pencil to its leading term.
\end{definition}
Clearly the transformation generated by the vector field 
$$d_1\left(-\sum_{i=1}^n\int c^i(u^i)u^i_x{\rm log}u^i_x\,dx\right)$$
is quasitrivial. In other words, in the semisimple case, all second order deformations are quasitrivial.
 In \cite{LZ2} Liu and Zhang proved that, in the scalar case, the deformations (if they exist) 
are quasitrivial at any order (an alternative independent proof was given later in \cite{B}). We will use this fact in subsequent section. The quasi-triviality of deformations of semisimple bi-Hamiltonian structures of hydrodynamic type was instead proved in \cite{DLZ}. 

\section{An alternative formalism}
We already pointed out that it is possible to compute the Schouten bracket of two bivectors using a different formalism, initially introduced by Dorfman, Gelfand \cite{GD}, Olver \cite{O} and further developed by Getzler \cite{G} and Barakat \cite{B}. On of the key advantages of this formalism is that it turns difficult and time consuming computations into extremely fast and straightforward calculations. 
Consider the graded algebra $\mathcal{A}:=\oplus_{k\in \mathbb{Z}} \mathcal{A}_{k}$ over the ring $C^{\mathbb{\infty}}(\mathbb{R}^n)$, where $\mathcal{A}=C^{\mathbb{\infty}}(\mathbb{R}^n)[u_{(1)}, u_{(2)}, \dots]$ is just the polynomial algebra with coefficients in $C^{\mathbb{\infty}}(\mathbb{R}^n)$ generated by countable generators $\{u_{(1)}, \dots, u_{(k)}, \dots\}$ with the grading induced by assigning $\mathrm{deg}(u^i_{(k)})=k$. 
On $\mathcal{A}$ it is defined a total derivative with respect to $x$: 
\begin{equation}\label{totalderivativex.eq}
\partial_x:=\sum_{i=1}^n \sum_{k=0}^{\infty}u^i_{(k+1)}\frac{\partial}{\partial u^i_{(k)}}
\end{equation}
and a variational derivative with respect to $u^i$ 
\begin{equation}\label{variationalu.eq}
\frac{\delta}{\delta u^i}:=\sum_{k\geq 0} (-\partial_x)^k \frac{\partial}{\partial u^i_{(k)}}
\end{equation}
Now instead of dealing with $\delta$-Dirac distributions and their derivatives, one introduces a polynomial algebra over anticommuting variables $\theta^i_{k}$ $i=1,\dots, n$, $k\in \mathbb{Z}_{\geq 0}$ that satisfies the following relations 
\begin{eqnarray}\label{theta.eq}
\theta^i_k\wedge \theta^j_l=- \theta^j_l \wedge \theta^i_k,\\
\partial_x \theta^i_k=\theta^i_{k+1},\\
\end{eqnarray}
where $\partial_x$ behaves like a derivation: $$\partial_x (\theta^i_k\wedge \theta^j_l)=\theta^i_{k+1}\wedge \theta^j_l+\theta^i_{k}\wedge \theta^j_{l+1}.$$
Formally, it is possible to express $\partial_x$ as a combination of partial derivatives with respect to $\theta^i_k$: 
\begin{equation}\label{partialtheta.eq}
\partial_x=\sum_{i=1}^n \sum_{k=0}^{\infty} \theta^i_{k+1}\frac{\partial}{\partial \theta^i_k}.
\end{equation}
To apply correctly formula \eqref{partialtheta.eq} obtaining results consistent with $\partial_x (\theta^i_k\wedge \theta^j_l)=\theta^i_{k+1}\wedge \theta^j_l+\theta^i_{k}\wedge \theta^j_{l+1}$, it is important to underline the following: when the operator $\frac{\partial}{\partial \theta^i_k}$ acts on an expression of the form $\theta^j_l\wedge \theta^i_k$ we need to bring the term $\theta^i_k$ in front using the anti-commutation rule $\theta^j_l\wedge \theta^i_k=-\theta^i_k\wedge\theta^j_l$ and only after that  apply the operator  $\frac{\partial}{\partial \theta^i_k}$. For instance, suppose that in the scalar case we want to compute $\partial (\theta_1\wedge \theta_3)=\theta_2\wedge \theta_3+\theta_1\wedge \theta_4$, using the equation \eqref{partialtheta.eq}. 
We have 
$$\sum_{k=0}^{\infty}\theta_{k+1} \frac{\partial}{\partial \theta_k} (\theta_1\wedge \theta_3)=$$
$$\theta_{2}\frac{\partial}{\partial \theta_1} (\theta_1\wedge \theta_3)+\theta_{4} \frac{\partial}{\partial \theta_3} (\theta_1\wedge \theta_3)=$$
$$\theta_{2}\wedge \theta_3-\theta_4\frac{\partial}{\partial \theta_3}(\theta_3\wedge \theta_1)=$$
$$\theta_{2}\wedge \theta_3 - \theta_4\wedge \theta_1=\theta_{2}\wedge \theta_3+\theta_1\wedge \theta_4.$$
From now on, when the operator $\frac{\partial}{\partial \theta^i_k}$ appears we will assume that this procedure has been enforced. 
In this way the total derivative $\partial_x$ is extended to $\mathcal{A}(\Theta):=\mathcal{A}[\theta^i_k, i=1, \dots, n, k\geq0]$:
\begin{equation}\label{partialextended.eq}
\partial_x=\sum_{i=1}^n \sum_{k=0}^{\infty}\left[ u^i_{(k+1)}\frac{\partial}{\partial u^i_{(k)}}+\theta^i_{k+1}\frac{\partial}{\partial \theta^i_k}\right]
\end{equation}
Analogously, we can consider the variational derivative with respect to $\theta^i$ as given by the following formula: 
\begin{equation}\label{variationaltheta.eq}
\frac{\delta}{\delta \theta^i}:=\sum_{k=0}^{\infty} (-\partial)^k \frac{\partial}{\partial \theta^i_{k}}.
\end{equation}
For instance, 
$$\frac{\delta}{\delta \theta^i}(\theta^i \wedge \theta^i_1)= \frac{\partial}{\partial \theta^i}(\theta^i \wedge \theta^i_1)-\partial_x(\frac{\partial}{\partial \theta^i_1}(\theta^i \wedge \theta^i_1))=$$
$$ \theta^i_1+\partial_x(\frac{\partial}{\partial \theta^i_1}(\theta^i_1 \wedge \theta^i))=2\theta^i_1.$$

As in the classical case, we have the following important lemma:
\begin{lemma}\label{variationofatotal}
The following identities hold: 
\begin{equation}\label{variationofatotal.eq}\frac{\delta}{\delta u^i}\partial_x=0, \quad \frac{\delta}{\delta \theta^i}\partial_x =0.\end{equation}
\end{lemma}
This lemma is a generalization of the well-known fact that to a Lagrangian function it is possible to add a closed form or a total derivative without affecting the equations of motion (Euler-Lagrange equations). 

\noindent\begin{proof}
The same proof of \cite{O2}, Theorem 4.7. applies to this case. 
\end{proof}

It is possible to introduce also higher-order variational derivatives with respect to $\theta^i_k$ and $u^i_{(k)}$: 
\begin{equation}\label{variationalhighertheta.eq}
\frac{\delta}{\delta \theta^i_k}:=\sum_{l=0}^{\infty} (-1)^l \binom{k+l}{k} \partial_x^l \frac{\partial}{\partial \theta^i_{k+l}},
\end{equation}
\begin{equation}\label{variationalhigheru.eq}
\frac{\delta}{\delta u^i_{(k)}}:=\sum_{l=0}^{\infty} (-1)^l \binom{k+l}{k} \partial_x^l \frac{\partial}{\partial u^i_{(k+l)}}
\end{equation}
The higher order variational derivatives are related to the ordinary partial derivatives through the following Lemma.
\begin{lemma}\label{ordinaryandvariational.th}
The following identities hold true: 
\begin{equation}\label{ordinaryvariationa1.eq}
\frac{\partial}{\partial \theta^i_k} =\sum_{j=k}^{\infty} \binom{j}{k} \partial^{j-k}_x \frac{\delta}{\delta \theta^i_j},
\end{equation}
where $ \frac{\delta}{\delta \theta^i_j}$ are given by \eqref{variationalhighertheta.eq} and 
\begin{equation}\label{ordinaryvariational2.eq}
\frac{\partial}{\partial u^i_{(k)}} =\sum_{j=k}^{\infty} \binom{j}{k} \partial^{j-k}_x \frac{\delta}{\delta u^i_{(j)}}, 
\end{equation}
where $\frac{\delta}{\delta u^i_{(j)}}$ are given by \eqref{variationalhigheru.eq}.
\end{lemma}
\begin{proof}
See Appendix
\end{proof}

An important identity relating higher variational derivatives is the following one: 
\begin{lemma}
For any differential polynomials $f, g\in \mathcal{A}$, the following identity holds:
\begin{equation}\label{highervariationalidentity.eq}
\sum_{j\geq 0 } \partial^j_x(f)\frac{\partial g}{\partial u_{(j)}}=\sum_{j\geq 0}\partial^j_x\left(f\frac{\delta g}{\delta u_{(j)}}\right)
\end{equation}
\end{lemma}
This Lemma is actually the definition of higher-variational derivatives as given in \cite{O2}. This Lemma holds true in a more general situation, where $f,g$ are not required to depend polynomially on the derivatives of $u$.

As we will see, this formalism has several advantages. First we need to recall how it is related to the construction of $k$-multivectors and evolutionary vector fields 
introduced before.
Given a $k$-multivector $P$ written in the Dubrovin-Zhang formalism, to re-write it in this formalism it is sufficient to substitute each occurrence of $\delta^{(k)}$ with $\theta_{k}$ and finally multiply by $\theta$ on the left. For instance, $\omega_1=\delta'(x-y)$ is written as $\omega_1=\theta\wedge \theta_1$, while $\omega_2=u\delta'(x-y)+\frac{1}{2}u_{x}\delta(x-y)$ is written as $\omega_2=\theta\wedge(u\theta_{1}+\frac{1}{2}u_{x}\theta)=u\theta\wedge \theta_{1}$. The same procedure applies in particular to evolutionary vector fields. Since an evolutionary vector field is written as $X=f\frac{\partial}{\partial u} +(\partial_x f)\frac{\partial}{\partial u_x}+...$, $f\in  \mathcal{A}$ in the Dubrovin-Zhang formalism, in this formalism the same vector field appears as $X=f\theta$. In the case of systems, if we are given a Poisson tensor of hydrodynamic type as $P^{ij}=g^{ij}\delta'(x-y)+\Gamma^{ij}_ku^k_x\delta(x-y)$, we can write it in terms of the anti-commuting variables $\theta^i$ as $P^{ij}=g^{ij}\theta^i \theta^j_1+\Gamma^{ij}_ku^k_x\theta^i \theta^j=g^{ij}\theta^i \theta^j_1$, since sum over $i, j$ is assumed and $\Gamma^{ij}$ is symmetric, while $\theta^i \theta^j$ is skew.

Thus using this formalism a $k$-multivector  $P$ is represented as a sum of terms of the form $f \theta^{i_1}_{l_1}\wedge\dots \wedge \theta^{i_k}_{l_k}$, where $f\in \mathcal{A}$. 
In particular, the Schouten bracket between a $k$-multivector $P$ and a $k'$-multivector $Q$ is a $(k+k'-1)$-multivector given by the following expression
\begin{equation}\label{schouten2}
[P,Q]=\sum_{i=1}^N \frac{\delta P}{\delta \theta^i} \frac{\delta Q}{\delta u^i}-(-1)^{k+1}\frac{\delta P}{\delta u^i} \frac{\delta Q}{\delta \theta^i}
\end{equation}

From \eqref{schouten2} and \eqref{variationofatotal.eq} we get immediately the following lemma: 
\begin{lemma}
Let $P$ and $Q$ be a $k$-multivector and $k'$-multivector respectively. Then 
$$[P, \partial Q]=0.$$
\end{lemma}

 Using this formalism, it might be useful to have a way to express $\frac{\delta}{\delta u^i}(f\theta^j_{p})$ using a formula in which the $\theta$-variables do not appear under an operator sign. This is provided by the following:
 \begin{lemma}\label{portarethetafuorilm}
 Let $f\in \mathcal{A}$ be homogenous of degree $k$. Then the following formula holds true:
 \begin{equation}\label{portarethetafuori1.eq}
 \frac{\delta}{\delta u^i}(f\theta^j_{p})=\sum_{l=0}^k (-1)^l  \frac{\delta}{\delta u^i_{(l)}}(f) \theta^j_{p+l},
 \end{equation}
 where moreover
 $$\frac{\delta}{\delta u^i_{(l)}}(f)=\sum_{h=0}^{k-l} (-1)^h\binom{l+h}{h}\partial^h\left(\frac{\partial f}{\partial u^i_{(l+h)}}\right).$$
 \end{lemma}
 \proof See Appendix.
 \endproof
 
Entirely similar formulas hold for more complicated expressions. 

Once we have an expression written using the anti-commutative variables $\theta^i$ and their derivatives, in order to revert to the Dubrovin-Zhang formalism it is necessary to apply a normalization operator $\mathcal{N}:=\sum_{i} \theta^i \frac{\delta}{\delta \theta^i}$ (introduced in \cite{B}) before deleting all the instances of $\theta^i$ appearing on the left and making the substitution $\theta_r \mapsto \delta^{(r)}$.
 Let's work out in more detail a simple example assuming that the target space is one dimensional. 
Consider a vector field $X=\partial^n_x f_n$, where $f_n(u)$ is an arbitrary function of $u$. We want to compute the Lie derivative of $\omega_1=\delta^{(1)}$ with respect to $X$. First we transform $\omega_1$ in the formalism with $\theta$'s, where it appears as $\omega_1=\theta \theta_1$, while $X=\partial^n_x f_n \theta$. Then we recall that the Poisson cohomology operator $d_1$ associated to $\omega_1$ is equal to : $d_1=2\theta_1 \frac{\delta}{\delta u}.$ So we have:
$$\mathrm{Lie}_X (\omega_1)=d_1(X)=2\theta_1 \frac{\delta}{\delta u}\left(\partial^n_x f \theta\right)=2\theta_1 \frac{\delta}{\delta u}\left((-1)^n f \theta_n\right),$$
where the last equality holds integrating by part inside the variational derivative and recalling that the variational derivative of a total derivative is identically zero. 
Thus we obtain:
$$\mathrm{Lie}_X (\omega_1)=(-1)^n2\frac{\partial f_n}{\partial u}\theta_1 \theta_n.$$
At this point, applying the normalization operator $\mathcal{N}$ we obtain:
$$\theta \frac{\delta}{\delta \theta}\left((-1)^n2\frac{\partial f_n}{\partial u}\theta_1 \theta_n \right)=$$
$$=(-1)^n\theta (-1)\partial_x\left( \frac{\partial}{\partial \theta_1}\left(2\frac{\partial f_n}{\partial u}\theta_1 \theta_n \right)\right)+(-1)^{2n}\theta\partial^n_x
\left( \frac{\partial}{\partial \theta_n}\left(2\frac{\partial f_n}{\partial u}\theta_1 \theta_n \right)\right)$$
$$=-(-1)^n \theta\partial_x\left(2\frac{\partial f_n}{\partial u}\theta_n  \right)-\theta \partial^n_x\left(   \frac{\partial}{\partial \theta_n}\left(\frac{\partial f_n}{\partial u}\theta_n \theta_1 \right)\right) $$
$$=-(-1)^n \theta\partial_x\left(2\frac{\partial f_n}{\partial u}\theta_n  \right)-\theta \partial^n_x\left(2\frac{\partial f_n}{\partial u} \theta_1 \right).$$
Now we can cancel the $\theta$ appearing on the left and substitute $\theta_1$ with $\delta^{(1)}(x-y)$ and $\theta_n$ with $\delta^{(n)}(x-y)$, obtaining
$$-(-1)^n\partial_x \left(2\frac{\partial f_n}{\partial u}\delta^{(n)}(x-y) \right)-\partial^n_x\left(2\frac{\partial f_n}{\partial u} \delta^{(1)}(x-y) \right).$$
This is exactly the formula for $d_1X_n$ appearing in Theorem \ref{def5th.th}

\section{The scalar case}
In the scalar case we have
$$\omega_1=f(u)\delta'(x-y)+\f{1}{2}f_x\delta(x-y)$$
and
$$\omega_2=g(u)\delta'(x-y)+\f{1}{2}g_x\delta(x-y).$$
Without loss of generality we can assume $f(u)=1$. For simplicity we 
 will consider  the special case $g(u)=2u$. In this case the pencil
$$\omega_2-\lambda\omega_1$$
is \emph{exact}. This means that there exist a vector field ($X=\f{1}{2}$) such that
$${\rm Lie}_X\omega_2=\omega_1,\qquad {\rm Lie}_X\omega_1=0.$$
For more about exact Poisson pencils and their deformations, see the next section.
Only two examples of deformations of $P_{\lambda}=\omega_2-\lambda\omega_1$ are known. One is the Poisson pencil of KdV which is
$$P_{\lambda}=\omega_2-\lambda\omega_1+c\delta'''(x-y),$$
 the second is the Poisson pencil of Camassa-Holm equation, that can be written in the form
$$\omega_2-\lambda\omega_1+P_{\epsilon}$$
with $P_{\epsilon}$ given by:
$$\sum_{n=1}\epsilon^{2n}\left[\d_x\,\left(u\,\delta^{(2n)}(x-y)\right)+\d^{2n}_x\,\left(u\,\delta'(x-y)\right)\right]+$$
$$+\sum_{n=1}\epsilon^{2n+1}\left[\d_x\left(u\, \delta^{(2n+1)}(x-y)\right)-\d_x^{2n+1}\left(u\, \delta'(x-y)\right)\right].$$

\begin{theorem}\label{def5th.th}
Up to the fifth order all deformations of the pencil 
\begin{equation*}
P_{\lambda}=\omega_2-\lambda\omega_1=2u(x)\delta'(x-y)+u_x\delta(x-y)-\lambda\delta'(x-y)
\end{equation*}
can be reduced, by the action of Miura group, to the following form
$$P^{\epsilon}_{\lambda}=\omega_2-\lambda\omega_1-\sum_{k=1}^5\epsilon^k d_1 X_k+\mathcal{O}(\epsilon^6)$$
where 
\begin{eqnarray*}
X_{n}&=&\left(\d_x^{n} f_{n}\right)\f{\d}{\d u}+\sum_{k\ge1}\epsilon^k\left(\d_x^{n+k} f_{n}\right)\f{\d}{\d u^{(k)}},\,n=1,\dots,5\\
d_1 X_n&=&-(-1)^n\d_x\left[\f{\d f}{\d u}\delta^{(n)}(x-y)\right]-\d_x^n\left[\f{\d f}{\d u}\delta'(x-y)\right]\\
\f{\d f_4}{\d u}&=&-\f{\d}{\d u}\left(\f{\d f_2}{\d u}\right)^2\\
\f{\d^2 f_{5}}{\d u^2}&=&-2\f{\d^2 f_{3}}{\d u^2}\f{\d^2 f_{2}}{\d u^2}
\end{eqnarray*}
and $f_2,f_3$ are arbitrary. The deformations are trivial if and only if $f_2=0$.
\end{theorem}

\noindent
This theorem has been proved in \cite{L}. For convenience of the reader, and as an example of the alternative formalism we outlined in the previous section, we will prove that the bivector defined above is Poisson up to terms of order $O(\epsilon^6)$ 

\noindent
\begin{proof}
First of all let us check that
$$[P^{\epsilon}_{\lambda},P^{\epsilon}_{\lambda}]=\mathcal{O}(\epsilon^6)$$
or, using the alternative formalism that
$$\{\hat{P}^{\epsilon}_{\lambda},\hat{P}^{\epsilon}_{\lambda}\}=\mathcal{O}(\epsilon^6),$$
(we denote with $\hat{P}$ the corresponding quantity in the alternative formalism introduced in Section 3).
This is equivalent to
$$2\{\omega_2,\widehat{d_1 X_{n}}\}+\sum_{k=1}^{n-1}\{\widehat{d_1 X_{2k}},\widehat{d_1 X_{2n-2k}}\},\qquad n=1,2.$$
Using the identities
\begin{eqnarray*}
&&\f{\d}{\d u_s}\d_x^n=\sum_{l=0}^n\binom{n}{l}\d_x^l\f{\d}{\d u_{s-n+l}}\\
&&\sum_{s=0}^n(-1)^s\binom{s}{k}\binom{n}{s}=(-1)^n\delta_{n,k}
\end{eqnarray*}
it is easy to prove that (from now on we will omit the symbol of wedge product among the anti-commuting variables $\theta$)
\begin{eqnarray}
&&\f{\delta \widehat{d_1 X_n}}{\delta u}=2(-1)^n\f{\d^2 f_n}{\d u^2}\theta_1\theta_n\\
&&\f{\delta \widehat{d_1 X_n}}{\delta\theta}=-2(-1)^n\d_x\left(\f{\d f_n}{\d u}\theta_n\right)-2\d_x^n\left(\f{\d f_n}{\d u}\theta_1\right).
\end{eqnarray}
Using the formulas above we obtain
\begin{eqnarray*}
&&2\{\omega_2,\widehat{d_1 X_{2n}}\}+\sum_{k=1}^{n-1}\{\hat{P}_{k},\hat{P}_{n-k}\}=\\
&&4\f{\d f_{n}}{\d u}\theta_0\theta_1\theta_{n+1}+4\d_x^{n}(\theta_0\theta_1)\f{\d f_{n}}{\d u}\theta_1+\\
&&+8\sum_{k=1}^{n-1}\f{\d f_{k}}{\d u}\f{\d^2 f_{n-k}}{\d u^2}\theta_1\theta_{n-k}\theta_{k+1}
+8\sum_{k=1}^{n-1}(-1)^{k}\d_x^{k}\left[\f{\d f_{k}}{\d u}\theta_1\right]\left[\f{\d^2 f_{n-k}}{\d u^2}\theta_1\theta_{n-k}\right]=0
\end{eqnarray*}
Taking into account  that
$$\d_x^{n}(\theta_0\theta_1)=\theta_0\theta_{n+1}+\sum_{k=1}^{\left[\f{n}{2}\right]}\left[\binom{n}{k}-\binom{n}{k-1}\right]\,\theta_k\theta_{n+1-k}$$
(where the square bracket denotes the integer part of the fraction) and dividing by 4, we obtain
\begin{eqnarray*}
&&\f{\d f_{n}}{\d u}\theta_0\theta_1\theta_{n+1}-\f{\d f_{n}}{\d u}\theta_0\theta_1\theta_{n+1}+\sum_{k=2}^{\left[\f{n}{2}\right]}\left[\binom{n}{k}-\binom{n}{k-1}\right]\,\f{\d f_{n}}{\d u}\theta_1\theta_k\theta_{n+1-k}+\\
&&+2\sum_{k=2}^{\left[\f{n}{2}\right]}
\left\{\f{\d f_{n-k}}{\d u}\f{\d^2 f_{k}}{\d u^2}-\f{\d f_{k-1}}{\d u}\f{\d^2 f_{n+1-k}}{\d u^2}\right\}\theta_1\theta_{k}\theta_{n+1-k}+\\
&&+2\sum_{k=1}^{n-1}(-1)^{n-k}\d_x^{n-k}\left[\f{\d f_{n-k}}{\d u}\theta_1\right]\left[\f{\d^2 f_{k}}{\d u^2}\theta_1\theta_{k}\right]
\end{eqnarray*}
that implies
\begin{eqnarray*}
&&\sum_{k=2}^{\left[\f{n}{2}\right]}
\left\{
\left[\binom{n}{k}-\binom{n}{k-1}\right]
\f{\d f_{n}}{\d u}+2\f{\d f_{n-k}}{\d u}\f{\d^2 f_{k}}{\d u^2}
-2\f{\d f_{k-1}}{\d u}\f{\d^2 f_{n+1-k}}{\d u^2}\right\}\theta_1\theta_{k}\theta_{n+1-k}+\\
\nonumber
&&+2\sum_{k=2}^{n-2}(-1)^{n-k}\d_x^{n-k}\left[\f{\d f_{n-k}}{\d u}\theta_1\right]\left[\f{\d^2 f_{k}}{\d u^2}\theta_1\theta_{k}\right]=0
\end{eqnarray*}
In the case $n=2$ the equation above is clearly satisfied for arbitrary
 $f_2$ e $f_3$. For $n=4$ we obtain
\begin{eqnarray*}
&&\left\{
\left[\binom{4}{2}-\binom{4}{1}\right]
\f{\d f_{4}}{\d u}+2\f{\d f_{2}}{\d u}\f{\d^2 f_{2}}{\d u^2}
\right\}\theta_1\theta_{2}\theta_{3}+2\d_x^{2}\left[\f{\d f_{2}}{\d u}\theta_1\right]\left[\f{\d^2 f_{2}}{\d u^2}\theta_1\theta_{2}\right]=\\
&&2\left\{\f{\d f_{4}}{\d u}+2\f{\d f_{2}}{\d u}\f{\d^2 f_{2}}{\d u^2}\right\}\theta_1\theta_{2}\theta_{3}=0
\end{eqnarray*}
that implies
$$\f{\d f_{4}}{\d u}=-\f{\d}{\d u}\left(\f{\d f_2}{\d u}\right)^2.$$
In the case $n=5$ we obtain
\begin{eqnarray*}
&&\left\{
\left[\binom{5}{2}-\binom{5}{1}\right]
\f{\d f_{5}}{\d u}+2\f{\d f_{3}}{\d u}\f{\d^2 f_{2}}{\d u^2}\right\}\theta_1\theta_{2}\theta_{4}+\\
&&+2\sum_{k=2}^{3}(-1)^{5-k}\d_x^{5-k}\left[\f{\d f_{5-k}}{\d u}\theta_1\right]\left[\f{\d^2 f_{k}}{\d u^2}\theta_1\theta_{k}\right]=\\
&&\left\{-5\d_x\left(\f{\d f_{5}}{\d u}\right)-10\f{\d^2 f_{3}}{\d u^2}\f{\d^2 f_{2}}{\d u^2}u_x\right\}\theta_1\theta_{2}\theta_{3}=0
\end{eqnarray*}  
that implies
$$\f{\d^2 f_{5}}{\d u^2}=-2\f{\d^2 f_{3}}{\d u^2}\f{\d^2 f_{2}}{\d u^2}.$$

\end{proof}

\noindent

Unfortunately, as we mentioned in the introduction,  it is not possible to extend the previous formulas to the case $n=6$ . Indeed in this case
 we obtain
\begin{eqnarray*}
&&\sum_{k=2}^{3}
\left\{
\left[\binom{6}{k}-\binom{6}{k-1}\right]
\f{\d f_{6}}{\d u}+2\f{\d f_{6-k}}{\d u}\f{\d^2 f_{k}}{\d u^2}
-2\f{\d f_{k-1}}{\d u}\f{\d^2 f_{7-k}}{\d u^2}\right\}\theta_1\theta_{k}\theta_{7-k}+\\
&&+2\sum_{k=2}^{4}(-1)^{6-k}\d_x^{6-k}\left[\f{\d f_{6-k}}{\d u}\theta_1\right]\left[\f{\d^2 f_{k}}{\d u^2}\theta_1\theta_{k}\right]=\\
&&\left\{9\f{\d f_{6}}{\d u}+4\f{\d f_{4}}{\d u}\f{\d^2 f_{2}}{\d u^2}\right\}\theta_1\theta_{2}\theta_{5}+\left\{5\f{\d f_{6}}{\d u}-4\f{\d f_{2}}{\d u}\f{\d^2 f_{4}}{\d u^2}\right\}\theta_1\theta_{3}\theta_{4}+\\
&&+\left\{4\d_x\left(\f{\d f_{4}}{\d u}\right)\f{\d^2 f_{2}}{\d u^2}\right\}
\theta_1\theta_{2}\theta_{4}+\left\{6\d_x^2\left(\f{\d f_{3}}{\d u}\right)\f{\d^2 f_{3}}{\d u^2}+12\d_x^2\left(\f{\d f_{4}}{\d u}\right)\f{\d^2 f_{2}}{\d u^2}\right\}\theta_1\theta_{2}\theta_{3}
\end{eqnarray*}
that implies
$$\f{\d f_{6}}{\d u}=-\f{\d}{\d u}\left(\f{\d f_2}{\d u}\f{\d f_4}{\d u}\right)$$
plus two additional conditions relating $f_2,f_3,f_4$ which are compatible only if
  $f_2(u)$ is a polynomial of degree 2.

As we will see the higher order deformations are much more complicated.

\section{Deformations of exact pencils in the scalar case}
In general, if $\omega_1$ is a Poisson structure, then $\omega_2:={\rm Lie}_X (\omega_1)$ is compatible with $\omega_1$ (in the sense that $[\omega_1, {\rm Lie}_X (\omega_1)]=0$) but it might fail to be Poisson itself. A simple sufficient condition ensuring that $\omega_2$ is Poisson is the notion of exact pencil, which was introduced in \cite{CMP}. 
\begin{definition}
Given a Poisson structure $\omega_1$ and a vector field $X$ such that $\omega_2:={\rm Lie}_X( \omega_1)\neq 0$ and ${\rm Lie}_X (\omega_2)=0$, we say that the pencil 
$\omega_{\lambda}:=\omega_1-\lambda\omega_2$ is an exact pencil.  
\end{definition}
Notice that the condition ${\rm Lie_X}(\omega_2)=0$ guarantees that ${\rm Lie}_X (\omega_1)$ is indeed a Poisson structure. In fact $0={\rm Lie}_X([\omega_2, \omega_1]=[{\rm Lie}_X(\omega_2), \omega_1]+[\omega_2, {\rm Lie}_X(\omega_1)]$, but ${\rm Lie}_X(\omega_2)=0$ by definition, so $[\omega_2, {\rm Lie}_X(\omega_1)]=[\omega_2, \omega_2]=0$. 

The following Lemma classifies exact Poisson pencils of hydrodynamic type in the scalar case:
\begin{lemma}\label{exactclassification}
In the scalar case, all exact Poisson pencils have the form 
$$\omega_{\lambda}=\omega_2-\lambda \omega_2=(au+b)\delta'(x-y)+\frac{1}{2}au_x\delta(x-y)-\lambda\delta'(x-y)=(au+b)\theta\theta_1-\lambda \theta\theta_1,$$
for arbitrary constants $a,b$. 
\end{lemma}
\begin{proof}
By triviality of the Poisson cohomology, we can always assume that one of the Poisson structures is $\partial_x$. Therefore, without loss of generality we can take $\omega_1:=\partial_x=\theta\theta_1$. 
We look for a vector field $X=f(u)\theta$, such that ${\rm Lie}_X(\omega_1)=d_1(X)=0$. Since $d_1=2\theta_1 \frac{\delta}{\delta u}$, we have that $d_1(X)=0$ is equivalent to requiring $f(u)$ to be a constant, call it $c$, so that $X=c\theta$. Given $\omega_2:=g(u)\theta \theta_1$, we search under which conditions on $g(u)$ $d_1(X)={\rm Lie}_X(\omega_2)$ is equal to $\omega_1$. We have
$$d_2(X)=\left(2g(u)\theta_1 +\frac{\partial g}{\partial u}u_x\theta \right)\frac{\delta}{\delta u}(c\theta)+\frac{\partial g}{\partial u}\theta\theta_1\frac{\delta}{\delta \theta}(c\theta)=c\frac{\partial g}{\partial u}\theta\theta_1.$$
Therefore $d_2(X)=\omega_1$ (up to the action of the normalization operator $\mathcal{N}$ which in this case acts as multiplication by $2$) if and only if $g(u)$ is at most affine in $u$: $g(u)=au+b$.  Moreover, since $c$ is an arbitrary constant different from zero, we can choose $c=\frac{1}{a}$, so that $d_2(X)$ is indeed $\omega_1$. Finally it is immediate to check that if $\omega_2$ has the form $(au+b)\theta\theta_1$, then $d_1(X)=0$ and 
$\omega_1$ is indeed Poisson. 
\end{proof}

Now we consider the deformation of a general pencil $\omega_2-\lambda \omega_1$ as follows:
\begin{equation}\label{deformationulterior}
\omega_2-\lambda\omega_1-d_1 X_{\epsilon}
=2g(u)\delta'(x-y)+\frac{\partial g}{\partial u}u_x\delta(x-y)-\lambda\delta'(x-y)-d_1 X_{\epsilon},\end{equation}
where $X_{\epsilon}=\sum_{k=1}^{\infty}\epsilon^k X_k$ and ${\rm deg}(X_k)=k$. (Notice that we have exchanged the names of $\omega_1$ and $\omega_2$, since we want to emphasize the role of $\delta'(x-y)$). 
By Miura quasi-triviality, the deformation vector field $X_{\epsilon}$ always exists, coming from Hamiltonian functionals
which are possibly not polynomials in the derivatives of $u$. 

It is not restrictive to assume that
 the odd powers in $\epsilon$ are missing, since this can be always achieved by performing a suitable Miura transformation.
 For istance, the third and fifth order deformations obtained in the previous section can be eliminated just by
 putting $f_3=0$, when $g(u)=u$. 
 For deformations of the form \eqref{deformationulterior} we have the following result:
\begin{theorem}
In the scalar case, the vector field $X_{\epsilon}$  is tangent to the symplectic leaves
 of $\omega_1$ if and only if the undeformed pencil $\omega_2-\lambda\omega_1$ is exact.
\end{theorem}

\proof 
The tangency  of $X_{\epsilon}$ to the symplectic leaves of $\omega_1$ is equivalent to impose the following condition  
$$\int_{S^1} X_{\epsilon}\f{\delta C}{\delta u}\,dx=0$$
for all the Casimirs $C$ of $\omega_1$.  On the other hand, the Casimirs of a Poisson bracket of hydrodynamic type
 are the integrals of the flat coordinates of the metric defining the bracket. In the case 
 of $\omega_1$ we have only one Casimir given by
$$C=\int_{S^1}u\,dx$$
and the tangency condition reads
$$\int_{S^1} X_{\epsilon}\,dx=0$$
which is equivalent to 
$$X_{\epsilon}=\d_x F_{\epsilon}$$
for a suitable differential polynomial $F_{\epsilon}$. Without loss of generality we can assume
$$X_{\epsilon}=\sum_{k=1}^{\infty}\epsilon^{2k}X_2^{(2k)}.$$
By the quasitriviality of deformations, we have that the deformed pencil:
\beq\label{defpen}
P_2-\lambda \omega_1=\omega_{\lambda}+{\rm Lie}_{X_{\epsilon}}\omega_1=\omega_{\lambda}+\sum_{k=1}^{\infty}\epsilon^{2k}P_2^{(2k)}
=\omega_{\lambda}+\sum_{k=1}^{\infty}\epsilon^{2k}{\rm Lie}_{X_2^{(2k)}}\omega_1
\eeq
can be reduced to its  dispersionless
 limit by iterating quasi Miura transformations (here $\omega_{\lambda}=\omega_2-\lambda\omega_1$). Following \cite{LZ2} we show how to construct such transformations.
 This will give us a crucial piece of information on the vector fields $X_2^{(2k)}$.

By hypothesis the vector field  $X^{(2)}_2$ satisfies the condition
$$d_1 d_2 X^{(2)}_2=0.$$
Moreover we have seen that it can be written as
\beq\label{lincomb1}
X^{(2)}_2=d_1 H_2^{(2)}-d_2 K_2^{(2)}
\eeq
for two suitable functionals $H^{(2)}_2$ and $K^{(2)}_2$. 
Let us consider the quasi Miura transformation generated by the vector field 
$$\epsilon^2 Z^{(2)}_2=-\epsilon^2 d_1 K_2^{(2)}.$$
Since 

\begin{eqnarray*}
{\rm Lie}_{Z^{(2)}_2}\omega_1&=&0\\
{\rm Lie}_{Z^{(2)}_2}\omega_2&=&-P^{(2)}_2=-{\rm Lie}_{X^{(2)}_2}\omega_1.
\end{eqnarray*}
this transformation does not modify
 $\omega_1$ while 
\begin{eqnarray*}
&&P_2\to\tilde{P}_2=P_2+\sum_{k=1}^{\infty}\f{\epsilon^{2k}}{k!}{\rm Lie}_{Z^{(2)}_2}^k P_2=\\
&&\omega_2+\epsilon^{4}\left(P^{(4)}_2+{\rm Lie}_{Z^{(2)}_2}P^{(2)}_2+\f{1}{2}{\rm Lie}_{Z^{(2)}_2}^2\omega_2\right)+\\
&&\epsilon^{6}\left(P^{(6)}_2+{\rm Lie}_{Z^{(2)}_2}P^{(4)}_2+\f{1}{2}{\rm Lie}_{Z^{(2)}_2}^2 P_2^{(2)}+\f{1}{6}{\rm Lie}_{Z^{(2)}_2}^3\omega_2\right)+\mathcal{O}(\epsilon^8)=\\
&&\omega_2+\epsilon^{4}\left(P^{(4)}_2-\f{1}{2}{\rm Lie}_{Z^{(2)}_2}^2\omega_2\right)
+\epsilon^{6}\left(P^{(6)}_2+{\rm Lie}_{Z^{(2)}_2}P^{(4)}_2-\f{1}{3}{\rm Lie}_{Z^{(2)}_2}^3\omega_2\right)
+\mathcal{O}(\epsilon^8)
\end{eqnarray*}
Notice that, using twice graded Jacobi identity, the second term of order $\mathcal{O}(\epsilon^4)$
 can be written as
$${\rm Lie}_{Z^{(2)}_2}^2\omega_2=[d_1 K^{(2)}_2,[d_1 K^{(2)}_2,\omega_2]]=-
[d_1 K^{(2)}_2,[d_2 K^{(2)}_2,\omega_1]]={\rm Lie}_{[d_1 K^{(2)}_2,d_2 K^{(2)}_2]}\omega_1$$
and therefore we can write
$$\tilde{P}_2=\omega_2+{\rm Lie}_{\tilde{X}^{(4)}_2}\omega_1+\mathcal{O}(\epsilon^6)$$
with
\beq\label{lincomb2}
\tilde{X}^{(4)}_2=X^{(4)}_2-\f{1}{2}[d_1 K^{(2)}_2,d_2 K^{(2)}_2]
\eeq
Moreover
$$d_1 d_2 \tilde{X}^{(4)}_2=0.$$
According to the main result of \cite{LZ2} (extended to the non scalar case in \cite{DLZ}) 
this implies that
$$\tilde{X}^{(4)}_2=d_1 H_2^{(4)}-d_2 K_2^{(4)}.$$

Let us consider now the quasi Miura transformation generated by the vector
 field 
$$\epsilon^4 Z^{(4)}_2=-\epsilon^4 d_1 K_2^{(4)}.$$ 
Since
\begin{eqnarray}
\label{Z4uno}
{\rm Lie}_{Z^{(4)}_2}\omega_1&=&0\\
\label{Z4due}
{\rm Lie}_{Z^{(4)}_2}\omega_2&=&-\left(P^{(4)}_2-\f{1}{2}{\rm Lie}_{Z^{(2)}_2}^2\omega_2\right)=-{\rm Lie}_{\tilde{X}^{(4)}_2}\omega_1,
\end{eqnarray}
this transformation does not modify $\omega_1$ while 
$$\tilde{P}_2\to\tilde{\tilde{P}}_2=
\omega_2+\epsilon^6\left(P^{(6)}_2
+{\rm Lie}_{Z^{(4)}_2}P^{(2)}_2+{\rm Lie}_{Z^{(2)}_2}P^{(4)}_2-\f{1}{3}{\rm Lie}_{Z^{(2)}_2}^3\omega_2\right)+\dots$$
Notice that, using twice graded Jacobi identity, the second and the third term in the term of order $\mathcal{O}(\epsilon^6)$ 
 can be written as
\begin{eqnarray*}
{\rm Lie}_{Z^{(4)}_2}{\rm Lie}_{Z^{(2)}_2}\omega_2&=&[d_1 K^{(4)}_2,[d_1 K^{(2)}_2,\omega_2]]=
-[d_1 K^{(4)}_2,[d_2 K^{(2)}_2,\omega_1]]={\rm Lie}_{[d_1 K^{(4)}_2,d_2 K^{(2)}_2]}\omega_1\\
{\rm Lie}_{Z^{(2)}_2}{\rm Lie}_{Z^{(4)}_2}\omega_2&=&[d_1 K^{(2)}_2,[d_1 K^{(4)}_2,\omega_2]]=
-[d_1 K^{(2)}_2,[d_2 K^{(4)}_2,\omega_1]]={\rm Lie}_{[d_1 K^{(2)}_2,d_2 K^{(4)}_2]}\omega_1\\
{\rm Lie}_{Z^{(2)}_2}^3\omega_2&=&[d_1 K^{(2)}_2,[[d_1 K^{(2)}_2,d_2 K^{(2)}_2],\omega_1]]=
-[[d_1 K^{(2)}_2,[d_1 K^{(2)}_2,d_2 K^{(2)}_2]],\omega_1]
\end{eqnarray*}
and therefore we can write
$$\tilde{\tilde{P}}_2=\omega_2+{\rm Lie}_{\tilde{X}^{(6)}_2}\omega_1+\mathcal{O}(\epsilon^8)$$
with
$$
\tilde{\tilde{X}}^{(6)}_2=X^{(6)}_2+[d_1 K^{(2)}_2,d_2 K^{(4)}_2]+[d_1 K^{(4)}_2,d_2 K^{(2)}_2]
-\f{1}{3}[[d_1 K^{(2)}_2,[d_1 K^{(2)}_2,d_2 K^{(2)}_2]].
$$
Moreover
$$d_1 d_2 \tilde{\tilde{X}}^{(6)}_2=0.$$
Again this implies that
$$\tilde{\tilde{X}}^{(6)}_2=d_1 H_2^{(6)}-d_2 K_2^{(6)}$$
The quasi Miura transformation generated by the vector
 field 
$$\epsilon^6 Z^{(6)}_2=-\epsilon^6 d_1 K_2^{(6)}$$
reduces the pencil to the form 
$$\tilde{\tilde{P}}_\lambda\to\tilde{\tilde{\tilde{P}}}_\lambda=
\omega_\lambda+\mathcal{O}(\epsilon^{8}).$$
The higher terms can be treated in a similar way. This is the procedure to construct the quasi-Miura transformation reducing the pencil to its dispersionless limit presented in \cite{LZ2,DLZ}. What is important for our pourposes is that the vector fields $X_2^{(2k)},\,k=4,5,\dots$ can always be written as linear combination of Hamiltonian vector fields (w.r.t. $\omega_1$ or $\omega_2$) and commutators (or iterated commutators) of Hamiltonian vector fields. 
 For instance
$$X^{(4)}_2=d_1 H_2^{(4)}-d_2 K_2^{(4)}+\f{1}{2}[d_1 K^{(2)}_2,d_2 K^{(2)}_2]$$ 
and
$$X^{(6)}_2=d_1 H_2^{(6)}-d_2 K_2^{(6)}-[d_1 K^{(2)}_2,d_2 K^{(4)}_2]-[d_1 K^{(4)}_2,d_2 K^{(2)}_2]
+\f{1}{3}[[d_1 K^{(2)}_2,[d_1 K^{(2)}_2,d_2 K^{(2)}_2]].$$

To conclude the first part of the proof we have to show that
\begin{itemize}
\item If two translation invariant vector fields are tangent to the symplectic leaves of $\omega_1$ the same
 is true for their commutator. 
\item  If $g(u)= au+b$ then the vectors fields $d_1 H$ and $d_2 K$ are tangent to the symplectic leaves of $\omega_1$ for any choice of the functionals $H$ and $K$.
\end{itemize}
Concerning the first point it follows immediately by the formula of commutator of two translation invariant vector
 fields 
$$X=X_0\f{\d}{\d u}+\sum_{k=1}^{\infty}\d_x^k X_0\f{\d}{\d u_{(k)}}$$
and 
$$Y=Y_0\f{\d}{\d u}+\sum_{k=1}^{\infty}\d_x^k Y_0\f{\d}{\d u_{(k)}}.$$
Indeed, intregrating by parts we have
\begin{eqnarray*}
\int_{S^1}[X,Y]_0\,dx&=&\int_{S^1}\sum_{k=0}^{\infty}\left[\d_x^k X_0\f{\d Y_0}{\d u_{(k)}}-\d_x^k Y_0\f{\d X_0}{\d u_{(k)}}\right]\,dx=\\
&&\int_{S^1}\left[X_0\sum_{k=0}^{\infty}(-1)^k\d_x^k\f{\d Y_0}{\d u_{(k)}}-Y_0\sum_{k=0}^{\infty}(-1)^k\d_x^k\f{\d X_0}{\d u_{(k)}}\right]\,dx=\\
&&\int_{S^1}\left[X_0\f{\delta Y_0}{\delta u}-Y_0\f{\delta X_0}{\delta u}\right]\,dx
\end{eqnarray*}
that vanishes since, by hypothesis 
$$\f{\delta Y_0}{\delta u}=\f{\delta X_0}{\delta u}=0.$$
\newline
Concerning the second point we observe that the tangency of $d_1 H$ is trivial while the tangency of $d_2 K$
 can be easily checked by straightforward computation:
\begin{eqnarray*}
&&\int_{S^1}\left(2g(u)\d_x+\frac{\partial g}{\partial u}u_x\right)\f{\delta K}{\delta u}\,dx=
\int_{S^1}\left[\d_x\left(2g(u)\f{\delta K}{\delta u}\right)-\frac{\partial g}{\partial u}u_x\f{\delta K}{\delta u}\right]\,dx=\\
&&-\int_{S^1}\sum_{k=0}^{\infty}(-1)^k\d^k_x\left(\f{\d K}{\d u^{(k)}}\right)u_x\frac{\partial g}{\partial u}\,dx.
\end{eqnarray*}
Indeed if 
$g(u)=au+b$, using repeated integration by parts, we have 
\begin{eqnarray*}
&&\int_{S^1}\sum_{k=0}^{\infty}(-1)^k\d^k_x\left(\f{\d K}{\d u^{(k)}}\right)u_x\frac{\partial g}{\partial u}\,dx
=\int_{S^1}\sum_{k=0}^{\infty}(-1)^k\d^k_x\left(\f{\d K}{\d u^{(k)}}\right)u_xa\,dx=\\ 
&&=a\int_{S^1}\sum_{k=0}^{\infty}\left(\f{\d K}{\d u^{(k)}}\right)u^{(k+1)}\,dx=a\int_{S^1}\partial_x(K)\,dx=0.
\end{eqnarray*} 
\newline
\newline
To conclude the proof we show that if $g(u)\ne au+b$ the vector field $X_2^{(2)}$ is no longer tangent
 to the symplectic leaves of $\omega_1$. Indeed, the deformations of the pencil
$$2g(u)\delta'(x-y)+g'\,u_x\delta(x-y)-\lambda\delta'(x-y),$$
up to the second order, are given by the following formula (see \cite{D} where also fourth order deformations are computed)
\begin{eqnarray*}
&&2g(u)\delta'(x-y)+g'\,u_x\delta(x-y)-\lambda\delta'(x-y)+\\
&&+\epsilon^2\left\{
\frac{cg'}{4}\,\delta'''(x-y)+\frac{3}{8}(cg')' u_x\,\delta''(x-y)+\right.\\
&&\left.+\left[\left(\frac{c''g'}{8}+\frac{c'g''}{3}+\frac{5cg'''}{24}\right)u_x^2 +
\left(\frac{c'g'}{8}+\frac{7cg''}{24}\right)u_{xx}\right]\delta'(x-y)+\right.\\
&&\left.+\left[\left(\frac{c''g''}{24}+\frac{c'g'''}{12}+\frac{cg^{(4)}}{24}\right)u_x^3 +
\frac{1}{6}(c'g'' + cg''') u_x u_{xx} +\frac{cg''}{12}u_{xxx}\right]\delta(x-y)\right\}\\
\end{eqnarray*}
The term in the bracket $\{\dots\}$ should be equal to 
$${\rm Lie}_{(Y(u)u_{xx}+Z(u)u_x^2)}\delta'(x-y)$$ 
for a suitable choice of $Y$ and $Z$. Taking into account the tangency condition which implies $Z=\frac{\partial Y}{\partial u}$, we would obtain
$$\{\dots\}=2Y\delta'''(x-y)+3Y_x\delta''(x-y)+Y_{xx}\delta'(x-y).$$
But this is impossibile due to the presence, in the left hand side, of the term
$$\left[\left(\frac{c''g''}{24}+\frac{c'g'''}{12}+\frac{cg^{(4)}}{24}\right)u_x^3 +
\frac{1}{6}(c'g'' + cg''') u_x u_{xx} +\frac{cg''}{12}u_{xxx}\right]\delta(x-y)$$
vanishing only if $g''=0$.


\endproof

\section{Deformations up to the eighth order}

In this section we compute deformations up to the eighth order and we show that there are no obstructions to the existence of  a polynomial  deformation up to that order. Such a result has been obtained taking full advantage of the computational capabilities of {\tt Maple} and of the statement of Theorem 16. 
\begin{theorem}
Up to the eighth order the deformations of the pencil
$$2u\delta'(x-y)+u_x\delta(x-y)-\lambda\delta'(x-y)$$
are unobstructed and 
can be reduced to the following form 
$$P^{\epsilon}_{\lambda}=\omega_2-\lambda\omega_1-\sum_{k=1}^4\epsilon^{2k} d_1 ( F_{2k})+\mathcal{O}(\epsilon^{10})$$
where the vector field $X_{\epsilon}=\sum_{k=1}^4\epsilon^{2k} ( F_{2k})$ generating the deformation has homogenous components given by 
\begin{eqnarray*}
F_{2}&=&\d^2_x f_2\\
F_{4}&=&\d^4_x f_4\\
F_{6}&=&\d_x\left(t_0 u_{xxxxx}+t_1 u_{xxxx}u_x+t_2 u_{xxx}u_{xx}+t_3 u_{xxx}u_x^2+t_4 u^2_{xx}u_x
 +t_5 u_{xx}u_{x}^3+\right.\\
&&\left.+t_6 u_x^5\right)\\
F_{8}&=&\d_x\left(r_0 u_{xxxxxxx}+r_1 u_{xxxxxx}u_x+r_2 u_{xxxxx}u_{xx}+r_3 u_{xxxxx}u_x^2+r_4 u_{xxxx}u_{xxx} +\right.\\
&&\left.+r_5 u_{xxxx}u_{xx}u_x+r_6 u_{xxxx}u_x^3+r_7 u_{xxx}^2 u_x+r_8 u_{xxx}u_{xx}^2+r_9 u_{xxx}u_{xx}u_x^2+\right.\\
&&\left.+r_{10} u_{xxx}u_x^4+r_{11} u_{xx}^3u_x+r_{12} u_{xx}^2 u_x^3+r_{13} u_{xx}u_x^5+r_{14} u_x^7\right)
\end{eqnarray*}
where
$$\f{\d f_4}{\d u}=-\f{\d}{\d u}\left(\f{\d f_2}{\d u}\right)^2,$$
\begin{eqnarray*}
t_0:=-\frac{1}{2}\, \left( {\frac{\d f_2}{\d u}} \right) ^{2}{\frac {\d^{3}f_2}{\d{u}^{3}}}
  - \left( {\frac {\d^{2}f_2}{\d{u}^{2}}}   \right) ^{2}{\frac {\d f_2}{\d u}},
\end{eqnarray*}
\begin{eqnarray*}
t_1:=\frac{1}{2}\,{\it t_2}  -\frac{1}{4}\, \left( {\frac {\d^{2}f_2}{\d
{u}^{2}}}   \right) ^{3}-\frac{3}{8}\,
 \left( {\frac {\d f_2}{\d u}}  
 \right) ^{2}{\frac {\d^{4}f_2}{\d{u}^{4}}}  -{\frac {19}{12}}\,\frac{\d f_2}{\d u}\,\frac {\d^{3}f_2}{\d{u}^{3}}\,\frac{\d^{2}f_2}{\d u^2},  
\end{eqnarray*}
\begin{eqnarray*}
t_3:&=&\frac{5}{6}\,t_4  -\frac{1}{6}\,{\frac {\d t_1}{\d u}}-\frac{1}{6}\,{\frac {\d t_2}{\d u}} -4\, \left( {\frac {\d^{2}f_2}{\d{u}^{2}}}   \right) ^{2}{\frac {\d^{3}f_2}{
\d{u}^{3}}}+\\  
&&-{\frac {23}{8}}\,\frac {\d f_2}{\d u} \, 
 \left( {\frac {\d^{3}f_2}{\d{u}^{3}}}\right) ^{2}-{\frac {9}{16}}\, \left( {\frac {\d f_2}{\d u}}   \right) ^{2}{\frac {\d^{5}f_2}{\d{
u}^{5}}}  -\frac{7}{2}\, 
\frac {\d f_2}{\d u}\,\frac {\d^{4}f_2}{\d{u}^{4}}\, \frac {\d^{2}f_2}{\d{u}^{2}}, 
\end{eqnarray*}
\begin{eqnarray*}
r_0:=-\frac{1}{6}\, \left( {\frac {\d f_2}{\d u}}\right)^{3}{\frac {\d^{4}f_2}{\d{u}^{4}}}
  - \left( {\frac {\d^{2}f_2}{\d{u}^{2}}}\right) ^{3}{\frac {\d f_2}{\d u}}  -\frac{3}{2}\, \frac {\d^{2}f_2}{\d{u}^{2}}\,  \left( {\frac{\d f_2}
{\d u}}\right) ^{2}{\frac {\d^{3}f_2}{\d{u}^{3}}}, 
\end{eqnarray*}
\begin{eqnarray*}
r_1&:=&-{\frac {205}{468}}\, \frac {\d^{2}f_2}{\d{u}^{2}}\,\left( {\frac {\d f_2}{\d u}} \right) ^{2}{\frac {\d^{4}f_2}{\d{u}^{4}}}  -\frac{1}{12}\, \left( {\frac {\d f_2}{\d
u}}   \right) ^{3}{\frac {\d^{5
}f_2}{\d{u}^{5}}}+\\  
&&-{\frac {83}{
234}}\, \left( {\frac {\d^{2}f_2}{\d{u}^{2}}} \right) ^{2} \,\frac {\d f_2}{\d u}\, {\frac {\d^{3}f_2}{\d{u}^{3}}}+
\frac{1}{26}\, \left( {\frac {\d^{2}f_2}{\d{u}^{2}}}   \right) ^{4}-{\frac {107
}{234}}\, \left( {\frac {\d^{3}f_2}{\d{u}^{3}}} \right) ^{2} \left( {\frac {\d f_2}{\d u}}
   \right) ^{2}+\\
&&+{\frac {16}{13}}\,{\it r_2}
  -{\frac {7}{13}}\,{\it r_4}  -{\frac {2}{39}}\,\frac {\d t_1}{\d u}\, {\frac {\d f_2}{\d u}}
-{\frac {2}{39}}\,\frac {\d t_2}{\d u}\, {\frac {\d f_2}{\d u}}
  +{\frac {2}{39}}\,t_4 {\frac {\d f_2}{\d u}}, 
\end{eqnarray*}
\begin{eqnarray*}
r_3&:=&\frac{7}{9}\,r_5 -\frac{773}{81}\,
 \frac{\d^{3}f_2}{\d{u}^{3}}\,\left(\frac {\d f_2}{\d u}\right)^{2}\,\frac {\d^{4}f_2}{\d{u}^{4}}-\frac {15715}{1944}\,\left(\frac{\d^2
 f_2}{\d{u}^{2}}\right) ^{3}\,
\frac{\d^{3}f_2}{\d{u}^{3}}\\
&&-\frac {33163}{7776}\,\frac{\d^{2}f_2}{\d{u}^{2}} \, \left(\frac{\d f_2}{\d u}\right)^{2}\frac{\d^{5}f_2}{\d{u}^{5}}-\frac{7}{9}\,r_7 -\frac{8}{27}\,\frac{\d r_2}{\d u}-\frac {35}{81}\,t_4\frac{\d^{2}f_2}{\d{u}^{2}}+\\
&&+\frac{7}{81}\,\frac{\d t_4}{\d u}\,\frac{\d f_2}{\d u}+\frac{35}{81}\,\frac{\d t_1}{\d u}\,\frac {\d^{2}f_2}{\d{u}^{2}}+\frac{35}{81}\,\frac{\d t_2}{\d u}\,\frac {\d^{2}f_2}{\d{u}^{2}}
-\frac{7}{81}\,\frac {\d^{2}t_1}{\d{u}^{2}}\,\frac{\d f_2}{\d u}+\\
&&-\frac{7}{81}\,\frac {\d^{2}t_2}{\d{u}^{2}}\,\frac{\d f_2}{\d u}-\frac{20249}{972}\, 
\frac{\d^{2}f_2}{\d{u}^{2}}\,
\frac{\d f_2}{\d u}\,
\left(\frac{\d^{3}f_2}{\d{u}^{3}}\right)^{2}
-\frac{97}{8}\, 
\left(\frac{\d^{2}f_2}{\d{u}^{2}}\right) ^{2}\,\frac {\d f_2}{\d u}\,\frac {\d^{4}f_2}{\d{u}^{4}}+\\
 &&-\frac {1003}{2592}\,
 \left(\frac{\d f_2}{\d u}\right) ^{3}\,\frac {\d^{6}f_2}{\d{u}^{6}} ,
\end{eqnarray*}
\newpage
\begin{eqnarray*}
r_4&:=&\frac {341}{108}\, \left(\frac {\d^{3}f_2}{\d{u}^{3}}\right) ^{2}\, 
\left(\frac {\d f_2}{\d u}\right) ^{2}+\frac{5}{3}\,r_2 +\frac {128}{27}\,\frac {\d^{2}f_2}{\d{u}^{2}} \,
\left( \frac{\d f_2}{\d u}\right)^{2}\,\frac{\d^{4}f_2}{\d{u}^{4}}+\\  
&&+\frac {65}{144}\, \left(\frac {\d f_2}{\d u}\right)^{3}\,\frac{\d^{5}f_2}{\d{u}^{5}}
+\frac {1175}{108}\,\left(\frac {\d^{2}f_2}{\d{u}^{2}}\right)^{2}\,
\frac {\d f_2}{\d u}\,\frac{\d^{3}f_2}{\d{u}^{3}}+ \left(\frac{\d^{2}f_2}{\d{u}^{2}}\right) ^{4}+\\
&&+\frac{1}{9}\,\frac{\d t_1}{\d u}\,\frac{\d f_2}{\d u}+\frac{1}{9}\,\frac{\d t_2}{\d u}\,\frac{\d f_2}{\d u}
 -\frac{1}{9}\,t_4\,\frac {\d f_2}{\d u},  
\end{eqnarray*}
\begin{eqnarray*}
r_5&:=&\frac {568}{45}\,\frac{\d^{3}f_2}{\d{u}^{3}}\, 
\left(\frac {\d f_2}{\d u}\right)^{2}\,\frac {\d^{4}f_2}{\d{u}^{4}} 
 +{\frac {5107}{540}}\,
 \left( {\frac {\d^{2}f_2}{\d{u}^{2}}}\right) ^{3}\,\frac{\d^{3}f_2}{\d{u}^{3}}+\\
  &&+\frac{991}{216}\,\frac{\d^{2}f_2}{\d{u}^{2}}\,
\left(\frac{\d f_2}{\d u}\right)^{2}\frac{\d^{5}f_2}{\d{u}^5} +\frac {11}{5}\,r_7
-\frac{3}{5}\,r_8
  +\frac{2}{3}\,{\frac {\d r_2}{\d u}}
  +{\frac {101}{90}}\,t_4\frac {\d^{2}f_2}{\d{u}^{2}}+\\
&&-{\frac{2}{45}}\,\frac {\d t_4}{\d u}\,{\frac {\d f_2}{\d u}}
 -{\frac {101}{90}}\,\frac {\d t_1}{\d u}\,{\frac {\d^{2}f_2}{\d{u}^{2}}}  
-{\frac {101}{90}}\,\frac {\d t_2}{\d u}\,{\frac {\d^{2}f_2}{\d{u}^{2}}}  
+{\frac {2}{45}}\,\frac {\d^{2}t_1}{\d{u}^{2}}\,{\frac {\d f_2}{\d u}}\\
&&+{\frac {2}{45}}\,\frac {\d^{2}t_2}{\d{u}^{2}}\,{\frac {\d f_2}{\d u}}
+{\frac {29279}{1080}}\,\frac{\d^{2}f_2}{\d{u}^{
2}}\,\frac {\d f_2}{\d u}\, 
\left( {\frac {\d^{3}f_2}{\d{u}^{3}}}\right)^{2}+
{\frac {247}{20}}\, 
\left(\frac {\d^{2}f_2}{\d{u}^{2}}\right) ^{2}\frac{\d f_2}{\d u}\,{\frac {\d^{4}f_2}{\d{u}
^{4}}}+\\  
&&+{\frac {1673}{3600}}\,
 \left( {\frac {\d f_2}{\d u}}\right) ^{3}\,{\frac {\d^{6}f_2}{\d{u}^{6}}},  
\end{eqnarray*}
\begin{eqnarray*}
r_6&:=&{\frac {10801}{972}}\, \left( {\frac {\d^{4}f_2}{\d{u}^{4}}}\right) ^{2} \left( {\frac {\d f_2}{\d u}
}   \right) ^{2}+{\frac {55271}{5832}}
\, \left( {\frac {\d^{2}f_2}{\d{u}^{2}}} \right) ^{3}{\frac {\d^{4}f_2}{\d{u}^{4}}}+\\
 &&+{\frac {286849}{23328}}\, \left( {\frac {\d^
{2}f_2}{\d{u}^{2}}}\right) ^{2}\,\frac {\d f_2}{\d u}\,{\frac {\d^{5}f_2}{\d{u}^{5}}}+{\frac {698971}{46656}}\, \frac {\d^{3}f_2}{\d{u}
^{3}}\,\left( {\frac {\d f_2}{\d u}}\right)^{2}{\frac {\d^{5}f_2}{\d{u}^{5}}}+\\
&&+{\frac {578641}
{11664}}\, \left( {\frac{\d^{2}f_2}{\d{u}^{2}}} \right) ^{2} \left( {\frac {\d^{3}f_2}{\d{u}^{3}}} \right) ^{2}+{\frac {39635}{1458}}\,
 \left( {\frac{\d^{3}f_2}{\d{u}^{3}}}\right)^{3}{\frac{\d f_2}{\d u}}+r_9+\\
&&+{\frac {99689}{23328}}\,
\frac {\d^{2}f_2}{\d{u}^{2}}\,\left( {\frac {\d f_2}{\d u}}\right)^{2}{\frac {\d^{6}f_2}{\d{u}^{6}}}
  -\frac{5}{3}\,r_{11} -
\frac {26}{27}\,{\frac {\d r_5}{\d u}} +{\frac {19}{54}}\,{\frac {\d r_7}{\d u}}+\\
&&+\frac{5}{9}\,{\frac {\d r_8}{\d u}}+{\frac{4}{3}}\,\frac {\d t_5}{\d u}\,{\frac {\d f_2}{\d u}}
-\frac{20}{3}\,t_6\frac {\d f_2}{\d u}+
{\frac {70}{
81}}\,{\frac {\d^{2}r_2}{\d{u}^{2}}}+\\
&&-{\frac {383}{486}}\,\frac {\d t_1}{\d u}\,{\frac {\d^{3}f_2}{\d{u}^{3}}}  
-{\frac {383}{486}}\,\frac {\d t_2}{\d u}\,{\frac {\d^{3}f_2}{\d{u}^{3}}}  
-{\frac {679}{486}}\,\frac {\d^{2}t_1}{\d{u}^{2}}\,\frac {\d^{2}f_2}{\d{u}^{2}}+\\ 
&&-{\frac {679}{486}}\,\frac {\d^{2}t_2}{\d{u}^{2}}\,\frac {\d^{2}f_2}{\d{u}^{2}} 
+\frac {14}{243}\,\frac {\d^{3}t_1}{\d{u}^{3}}\,{\frac {\d f_2}{\d u}}
+\frac {14}{243}\,\frac {\d^{3}t_2}{\d{u}^{3}}\,{\frac {\d f_2}{\d u}}+\\
&&+{\frac {383}{486}}\,t_4\, {\frac {\d^{3}f_2}{\d{u}^{3}}} 
+{\frac {679}{486}}\,\frac {\d t_4}{\d u}\,{\frac {\d^{2}f_2}{\d{u}^{2}}} 
-{\frac {222}{243}}\,\frac{\d^{2}t_4}{\d{u}^{2}}\, {
\frac {\d f_2}{\d u}}+\\
&&+{\frac{
502015}{5832}}\,\frac{\d^{3}f_2}{\d{u}^{3}}\,\frac {\d f_2}{\d u}\,\frac {\d^{4}f_2}{\d{u}^{4}}\, {\frac {\d^{2}f_2}
{\d{u}^{2}}}+{\frac {5983}{
15552}}\, \left( {\frac {\d f_2}{\d u}}\right)^{3}{\frac {\d^{7}f_2}{\d{u}^{7}}},
 \end{eqnarray*} 
\begin{eqnarray*}
r_{10}&:=&+{\frac {20}{39}}\,\frac {\d t_1}{\d u}\,{\frac {\d^{4}f_2}{\d{u}^{4}}}
+{\frac {20}{39}}\,\frac {\d t_2}{\d u}\,{\frac {\d^{4}f_2}{\d{u}^{4}}}
  -{\frac {173507}{2808}}\,\frac {\d^{
3}f_2}{\d{u}^{3}}\,\frac {\d f_2}{\d u}\,\frac {\d^{5}f_2}{\d{u}^{5}}\,{\frac {\d^{2}f_2}{\d{u}^{2}}}+\\
 &&-{\frac {20}{39}}\,t_4\,{\frac {\d^{4}f_2}{\d{u}^{4}}} -{\frac {209}{1560}}\, \left( {\frac {\d f_2}{\d u}}
   \right) ^{3}{\frac {\d^{8}f_2}{\d{u}^{8}
}}-{\frac {12811}{312}}\, \left( {
\frac {\d^{3}f_2}{\d{u}^{3}}}  
 \right) ^{3}{\frac {\d^{2}f_2}{\d{u}^{2}}}+\\ 
&&+{\frac {5}{3}}\,\frac {\d^{2}f_2}{\d{u}^{2}}\,{\frac {\d t_5}{\d u}}
-{\frac {5}{9}}\,\frac {\d^{2}f_2}{\d{u}^{2}}\,{\frac {\d^2 t_4}{\d u^2}}
-{\frac {9907}{468}}\, \left( {\frac {\d^{
2}f_2}{\d{u}^{2}}}   \right) ^{3}{
\frac {\d^{5}f_2}{\d{u}^{5}}}  
+{\frac {35}{39}}\,\frac {\d^{2}t_5}{\d{u}^{2}}\, {\frac {\d f_2}{\d u}}
+\\
&& -\frac{25}{3}\,\frac {\d^{2}f_2}{\d{u}^{2}}\,t_6
-{\frac {5}{78}}\, \frac {\d t_4}{\d u}\,{\frac {\d^{3}f_2}{\d{u}^{3}}} -{\frac {66479}{702}}\,
\frac {\d^{3}f_2}{\d{u}^{3}}\, \left( {\frac {\d^{2}f_2}{\d{u}^{2}}}\right) ^{2}{\frac {\d^{4}f_2}{\d{u}^{4}}}\\
&&  -{\frac {32501}{2340}}\, \left( {\frac {\d^{2}f_2}{\d{u}^{2}}}   \right) ^{2}\,
 \frac {\d f_2}{\d u} \,
  {\frac {\d^{6}f_2}{\d{u}^{6}}} -{\frac {1323}{520}}\,\frac {\d^{2}f_2}{\d{u}^{2}
}\, \left( {\frac {\d f_2}{\d u}} \right) ^{2}{\frac {\d^{7}f_2}{\d{
u}^{7}}}+\\  
&&+{\frac {5}{78}}\,\frac {\d^{2}t_1}{\d{u}^{2}}\,{\frac {\d^{3}f_2}{\d{u}^{3}}}
+{\frac {5}{78}}\,\frac {\d^{2}t_2}{\d{u}^{2}}\,{\frac {\d^{3}f_2}{\d{u}^{3}}}
-{\frac {7353}{1040}}\,\frac {\d^{3}f_2}{\d{u}^{3}}\, \left( {\frac {\d f_2}
{\d u}}\right) ^{2}{\frac {\d^{6}f_2}{\d{u}^{6}}}+\\
&&-{\frac {557}
{52}}\,\frac {\d^{4}f_2}{\d{u}^{4}}\,  \left( {\frac {\d f_2}{\d u}}\right) ^{2}{\frac {\d^{5}f_2}{\d{u}^{5}}}-{\frac {12094}{351}}\, \left( {\frac {\d^{4}f_2}{\d{u}^{4}}}   \right) ^{2}
 \,\frac {\d f_2}{\d u}\, {\frac {\d^{2}f_2}{\d{u}^{2}}}\\
&&-{\frac {35393}{702}}\, \left( {\frac {\d^{3}f_2}{\d{u}^{3
}}}   \right) ^{2} \,\frac {\d f_2}{\d
u}\, {\frac {\d^{4}f_2}{\d
{u}^{4}}}  +{\frac {17}{26}}\,
{\frac {\d r_9}{\d u}}  -{\frac {55}{
39}}\,{\frac {\d r_{11}}{\d u}} +{
\frac {5}{13}}\,r_{12}+\\  
&&-{\frac {35}{117}}\,\frac {\d^{3}t_4}{\d{u}^{3}}\,{\frac {\d f_2}{\d u}}+{\frac {14}{13}}\,{\frac {\d^{2}r_8}{\d{u}^{2}}}  -{\frac {20}{13}}\,{
\frac {\d^{2}r_7}{\d{u}^{2}}}-\frac {175}{39}\,\frac {\d t_6}{\d u}\, {\frac {\d f_2}{\d u}}. 
\end{eqnarray*}
All the remaining parameters, namely $t_2,t_4,t_5,t_6$ and $r_2,r_7,r_8,r_9,r_{11},r_{12},r_{13},r_{14}$ are free and can be chosen arbitrarily.
\end{theorem}

Using the freedom in the choice of the parameters we can simplify the previous expression of $X_{\epsilon}$ and of the corresponding deformation. For instance, up to the sixth order 
 in $\epsilon$ we have the following result, where in order to have a simpler expression, the free 
parameter $f_3$ appearing in Theorem 13 has been fixed to 
0.
\begin{theorem}
Up  to Miura transformations, the deformations of the pencil $P_{\lambda}=u\delta^{1}+\frac{1}{2}u_{(1)}\delta-\lambda \delta^{(1)}$
can be reduced to the following form:
\begin{equation}\label{main.eq}
\begin{split}
P^{\epsilon}_{\lambda}=P_{\lambda}-\epsilon^2\left\{\partial^2_x\left(c_2\delta^{(1)}(x-y)\right)+c_2 \delta^{(3)}(x-y)+(\partial_x c_2)\delta^{(2)}(x-y)\right\}\\
-\epsilon^4 \left\{\partial^4_x \left(c_4 \delta^{(1)}(x-y)\right)+c_4\delta^{(5)}(x-y)+(\partial_x c_4)\delta^{(4)}(x-y)\right\}\\
-\epsilon^6\left\{\partial^6_x\left(c_6\delta^{(1)}(x-y)\right)+c_6 \delta^{(7)}(x-y)+(\partial_x c_6)\delta^{(6)}(x-y)\right\}\\
+\epsilon^6\left\{h \delta^{(3)}(x-y)+(\partial_x h) \delta^{(2)}(x-y)+\partial^2_x\left(h\delta^{(1)}(x-y)\right)\right\}\\
+\epsilon^6\left\{\partial^3_x\left((\partial^2_x g) \delta^{(2)}(x-y)\right)+\partial_x\left((\partial^3_x g)\delta^{(3)}(x-y)\right)+(\partial^2_x g)\delta^{(5)}(x-y)+(\partial^3_x g)\delta^{(4)}(x-y)\right\},
\end{split}
\end{equation}
where 
$$c_2=\f{\d f_2}{\d u},$$ 
$c_4$ and $c_6$ are related to $c_2$ via the following equations:
\begin{equation}\label{c4asc2.eq}
c_4=\frac{\d f_4}{\d u}=-\frac{\d}{\d u}(c_2)^2,
\end{equation}
\begin{equation}\label{c6asc2.eq}
c_6=-\frac{1}{2}\frac{\d}{\d u}\left( c^2_2\; \frac{\d c_2}{\d u} \right),
\end{equation}
while $g$ is given by 
\begin{equation}\label{g.eq}
g=\frac{1}{2}\mathlarger{\int} \left\{\frac{3}{2}c_2^2\;  \frac{\d^3 c_2}{\d u^3}+\left(\frac{\d c_2}{\d u}\right)^3+\frac{19}{3}c_2 \; \frac{\d^2 c_2}{\d u^2}\;\frac{\d  c_2}{\d u} \right\}\; du
\end{equation}
and $h:=h_1+h_2+h_3+h_4$ and the $h_i$'s have the following expression:

\begin{equation}
h_1=\,u_{xx}^{2}\left(\frac {97}{60}c_2 
\left(\frac{\partial^2 c_2}{\partial u^2}\right)^{2}+\frac{8}{3}\,\left(\frac{\partial c_2}{\partial u}\right)^{2}\frac{\partial^2 c_2}{\partial u^2} 
+\frac {21}{40}\, c_2^2\frac {\partial^4 c_2}{\partial u^4}+\frac {49}{15}\,c_2  
\left( \frac{\partial^3 c_2}{\partial u^3} \right) \frac{\partial c_2}{\partial u}\right)
\end{equation}
\begin{equation}
h_2=\,u_x^4 \left(
\begin{split}
\frac {254}{3}
\left(\frac{\partial c_2}{\partial u}\right) ^{2}\frac {\partial^4 c_2}{\partial u^4} +\frac {17}{5}\, \left(c_2\right)^{2}\frac {\partial^6 c_2}{\partial u^6} 
 +\frac{176}{3}\,c_2 \left(\frac {\partial^3 c_2}{\partial u^3}\right) ^{2}\\ +\frac {4018}{45}\, c_2 
\left( \frac {\partial ^4 c_2}{\partial u^4}\right) \frac{\partial^2 c_2}{\partial u^2}
+\frac {1684}{45}\,c_2 \frac{\partial^5 c_2}{\partial u^5} \frac{\partial c_2}{\partial u}+\frac {14512}{45}\,\left( \frac{\partial c_2}{\partial u}\right)  \left(  \frac{\partial^2 c_2}{\partial u^2}\right) \frac{\partial^3 c_2}{\partial u^3}
\end{split}
\right)
\end{equation}
\begin{equation}
h_3=u_{xxx}u_{x}\left( \frac{3}{10}\, c_2^2\frac{\partial^4 c_2}{\partial u^4} +\frac{2}{3}\, \left(\frac{\partial c_2 }{\partial u} \right) ^{2}\frac {\partial^2 c_2}{\partial u^2} +\frac{1}{15}\,c_2 \left( \frac{\partial^2 c_2}{\partial u^2}\right) ^{2}+\frac {28}{15}\,c_2 \left(\frac{\partial^3 c_2}{\partial u^3} \right) \frac{\partial c_2}{\partial u}\right)
\end{equation}
\begin{equation}
h_4=u_{xx}u_{x}^{2}\left(
\begin{split} \frac{139}{10}\, 
\left(\frac{\partial c_2}{\partial u} \right)  \left(  \frac{\partial^2 c_2}{\partial u^2} \right) ^{2}+\frac {178}{15}\,
\left(\frac{\partial c_2}{\partial u}\right) ^{2} \frac{\partial^3 c_2}{\partial u^3} 
+\frac {21}{20}\, c_2^2 \frac{\partial^5 c_2}{\partial u^5}\\
+\frac {259}{30}\,c_2  
\left(  \frac{\partial^4 c_2}{\partial u^4} \right)  \frac{\partial c_2}{\partial u} +13\,c_2 
\left( \frac{\partial^2 c_2}{\partial u^2} \right)  
\frac{\partial^3 c_2}{\partial u^3}
\end{split}
\right).
\end{equation}
\end{theorem}
Let us observe then, that although deformations up to the fourth order in $\epsilon$ appear to have a quite regular structure, the pattern is broken already at the sixth order. 

\section{Conclusions}
In this paper we proved that the linear relations on the unknown coefficients of the vector fields $X_{2k}$ generating the deformations, obtained by solving recursively \eqref{rec2} have a geometric meaning in the case of the pencil $\omega_{\lambda}$ we considered. The linear relations entail the tangency of $X_{2k}$ to the symplectic leaves of $\omega_1$ and this geometric property is equivalent to the exactness of the pencil $\omega_{\lambda}$. 
How much of this extends to the case of systems is not clear. Although an essential ingredient in our proof, namely the Miura-quasi triviality, has been established also for systems of hydrodynamic type, preliminary computations show that this is not enough to relate the exactness of the pencil to the tangency of the vector fields generating the deformation. This will be explored elsewhere.

Our computation about deformations up to the eighth order shows, unfortunately, that it is not feasible to guess a general formula for either the vector fields $X_{2k}$ or the Poisson structures, even though deformations up to the fourth order might have suggested otherwise. On the other hand, the formulas we have found might turn out to be useful in guessing general formulas not for arbitrary central invariants, but for specific ones. Indeed for a specific choice of central invariant, the formulas simply drastically (even to the eighth order) and some patterns start to emerge. It is seems therefore reasonable that using these formulas, using a specific central invariant, one might be able to guess a general structure (for that specific central invariant) and construct the associated new integrable PDE. 

\bigskip

{\bf Acknowledgements} The authors gratefully acknowledge partial support from the University of Toledo via the  URAF grant  ``Lax equations, integrable systems and application". We also want to thank the Department of Mathematics of the University of Milano - Bicocca and of the University of Toledo for the supportive environment they provided and for their hospitality while this work was completed. One of the authors (P. Lorenzoni) warmly thanks B. Dubrovin for having introduced him to the problem of bi-Hamiltonian deformations of PDEs of hydrodynamic type during his PhD thesis. 

\section{Appendix}
Proof of Lemma 9
\begin{proof}
Since the proof is based on combinatorial identities, we focus on the proof of \eqref{ordinaryvariationa1.eq}, since the proof of \eqref{ordinaryvariational2.eq} is entirely similar. Using direct substitution of \eqref{variationalhighertheta.eq} in \eqref{ordinaryvariationa1.eq} we have: 
$$\frac{\partial}{\partial \theta^i_k} =\sum_{j=k}^{\infty} \binom{j}{k} \partial^{j-k}_x \sum_{l=0}^{\infty} \binom{j+l}{j}\partial^l_x\frac{\partial}{\partial \theta^i_{l+j}}=$$
$$\sum_{j=k}^{\infty}\sum_{l=0}^{\infty} (-1)^l\binom{j}{k}\binom{j+l}{j}\partial^{j+l-k}_x\frac{\partial}{\partial \theta^i_{l+j}}.$$
Setting $q=l+j$ we get that the previous expression is equal to: 
$$\sum_{j=k}^{\infty}\sum_{q=j}^{\infty}(-1)^{q-j}\binom{j}{k}\binom{q}{j}\partial^{q-k} \frac{\partial}{\partial \theta^i_{q}} =
\sum_{q=k}^{\infty}\sum_{j=k}^{q}(-1)^{q-j}\binom{j}{k}\binom{q}{j}\partial^{q-k} \frac{\partial}{\partial \theta^i_{q}},$$
where this last equality comes from exchanging vertical with horizontal summation in a lattice. 
Finally using the identity 
$$\sum_{j=k}^q (-1)^j \binom{j}{k}\binom{q}{j}=(-1)^q \delta^q_{k},$$
where $\delta^q_{k}$ is Kronecker delta, 
we obtain that the last expression is equal to
$$(-1)^q\sum_{q=k}^{\infty}(-1)^q \delta^q_{k}\partial^{q-k} \frac{\partial}{\partial \theta^i_{q}}= \frac{\partial}{\partial \theta^i_{k}},$$
thus proving the identity. 
The proof of formula \eqref{ordinaryvariational2.eq} is completely analogous. 
\end{proof}

Proof of Lemma 12 
 \proof
 By definition we have
 $$ \frac{\delta}{\delta u^i}(f\theta^j_{p})=\sum_{r=0}^{\infty}(-1)^r \partial^r\left( \frac{\partial}{\partial u^i_{(r)}}(f\theta^j_{p})\right)$$
 $$=\sum_{r=0}^{\infty}(-1)^r \partial^r\left( \frac{\partial f}{\partial u^i_{(r)}}\theta^j_{p}\right)=\sum_{r=0}^k(-1)^r \partial^r\left( \frac{\partial f}{\partial u^i_{(r)}}\theta^j_{p}\right),$$
 where the last equality is due to the fact that $f$ is homogenous of degree $k$. Now we have
 $$\sum_{r=0}^k(-1)^r \partial^r\left( \frac{\partial f}{\partial u^i_{(r)}}\theta^j_{p}\right)=\sum_{r=0}^k(-1)^r \sum_{h=0}^r\binom{r}{h}\partial^h\left( \frac{\partial f}{\partial u^i_{(r)}}  \right) \theta^j_{p+r-h}.$$
 Setting $r-h=l$ and remembering that $f$ is homogenous of degree $k$, the previous sum can be re-written as
 $$\sum_{l=0}^k \sum_{h=0}^{k-l}(-1)^{l+h}\binom{l+h}{h}\partial^h\left( \frac{\partial f}{\partial u^i_{(l+h)}}  \right) \theta^j_{p+l}=$$
$$=\sum_{l=0}^k (-1)^l\left[\sum_{h=0}^{k-l}(-1)^{h}\binom{l+h}{h}\partial^h\left( \frac{\partial f}{\partial u^i_{(l+h)}}  \right) \right]\theta^j_{p+l} $$ 
$$=\sum_{l=0}^k (-1)^l \frac{\delta}{\delta u^i_{(l)}}(f)\; \theta^j_{p+l},$$
where the last equality holds by definition of higher order variational derivative (see \eqref{variationalhigheru.eq}) and the fact that $f$ is homogenous of degree $k$. 
 \endproof

\bibliographystyle{plain}

\end{document}